\documentclass{article}

\usepackage[utf8]{inputenc}
\usepackage{amsmath, amssymb, amsthm}
\usepackage{mathtools}
\usepackage{fullpage}
\usepackage{comment}

\usepackage[
backend=biber,
style=alphabetic,
citestyle=alphabetic,
maxalphanames=4,
maxcitenames=99,
mincitenames=98,
maxbibnames=99,
giveninits=true,
]{biblatex}
\usepackage{todonotes}
\presetkeys{todonotes}{}{}

\usepackage{enumitem}
\setlist[itemize,enumerate]{
topsep=0.4em,
itemsep=0.2em,
parsep=0em}

\setlist[itemize,enumerate,2]{
topsep=0.2em
}

\usepackage[ruled]{algorithm2e}
\usepackage{setspace}
\usepackage{xspace}
\usepackage{array}

\usepackage{hyperref}
\hypersetup{colorlinks, linkcolor=darkblue, citecolor=darkgreen, urlcolor=darkblue}

\usepackage[nameinlink, capitalise]{cleveref}

\usepackage{complexity}

\graphicspath{{graphics/}}

\definecolor{darkblue}{rgb}{0,0,0.38}
\definecolor{darkred}{rgb}{0.8,0,0}
\definecolor{darkgreen}{rgb}{0.1,0.35,0}

\newcommand{\AUXPROBLEM}{{\ensuremath{\operatorname{\mathcal F - CP}}}\xspace}

\newcommand{\F}{\Flower}
\newcommand{\discup}{\mathbin{\dot\cup}}
\DeclareMathOperator{\argmax}{argmax}
\DeclareMathOperator{\argmin}{argmin}
\DeclareMathOperator{\gain}{gain}
\DeclareMathOperator{\Core}{Core}
\DeclareMathOperator{\Cluster}{Cluster}
\DeclareMathOperator{\Flower}{Flower}

\DeclareMathOperator{\diam}{diam}
\DeclareMathOperator{\supp}{supp}
\DeclareMathOperator{\rk}{rk}

\newtheorem{theorem}{Theorem}
\newtheorem{definition}[theorem]{Definition}
\newtheorem{lemma}[theorem]{Lemma}
\newtheorem{claim}[theorem]{Claim}

\title{Techniques for Generalized Colorful \texorpdfstring{$k$}{k}-Center Problems}
\author{%
Georg Anegg\thanks{ETH Zurich. Email: ganegg@ethz.ch. Research supported in part by Swiss National Science Foundation grant number 200021\_184622.} \and%
Laura Vargas Koch\thanks{ETH Zurich, Universidad de Chile. Email: lvargas@ethz.ch} \and
Rico Zenklusen\thanks{ETH Zurich. Email: ricoz@ethz.ch. Research supported in part by Swiss National Science Foundation grant number 200021\_184622. This project has received funding from the European Research Council (ERC) under the European Union's Horizon 2020 research and innovation programme (grant agreement No 817750).}%
}
\date{}

\begin{document}

\maketitle

\begin{abstract}
Fair clustering enjoyed a surge of interest recently.
One appealing way of integrating fairness aspects into classical clustering problems is by introducing multiple covering constraints.
This is a natural generalization of the robust (or outlier) setting, which has been studied extensively and is amenable to a variety of classic algorithmic techniques.
In contrast, for the case of multiple covering constraints (the so-called colorful setting), specialized techniques have only been developed recently for $k$-Center clustering variants, which is also the focus of this paper. 

While prior techniques assume covering constraints on the clients, they do not address additional constraints on the facilities, which has been extensively studied in non-colorful settings.
In this paper, we present a quite versatile framework to deal with various constraints on the facilities in the colorful setting, by combining ideas from the iterative greedy procedure for Colorful $k$-Center by Inamdar and Varadarajan with new ingredients.
To exemplify our framework, we show how it leads, for a constant number $\gamma$ of colors, to the first constant-factor approximations for both Colorful Matroid Supplier with respect to a linear matroid and Colorful Knapsack Supplier.
In both cases, we readily get an $O(2^\gamma)$-approximation.

Moreover, for Colorful Knapsack Supplier, we show that it is possible to obtain constant approximation guarantees that are independent of the number of colors $\gamma$, as long as $\gamma=O(1)$, which is needed to obtain a polynomial running time.
More precisely, we obtain a $7$-approximation by extending a technique recently introduced by Jia, Sheth, and Svensson for Colorful $k$-Center.
\end{abstract}

\begin{tikzpicture}[overlay, remember picture, shift = {(current page.south east)}]
\begin{scope}[shift={(-1.1,2.5)}]
\def\hd{2.5}
\node at (-2.15*\hd,0) {\includegraphics[height=0.7cm]{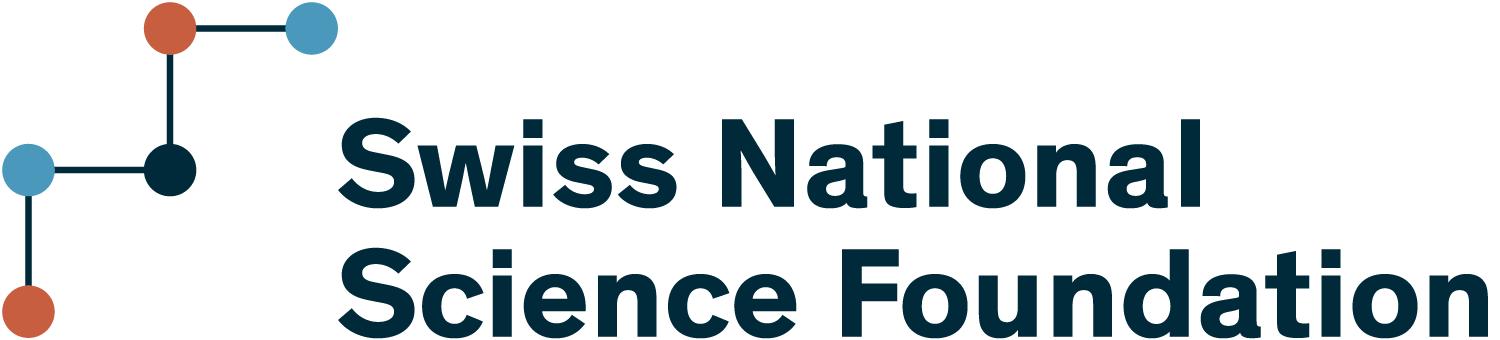}};
\node at (-\hd,0) {\includegraphics[height=1.0cm]{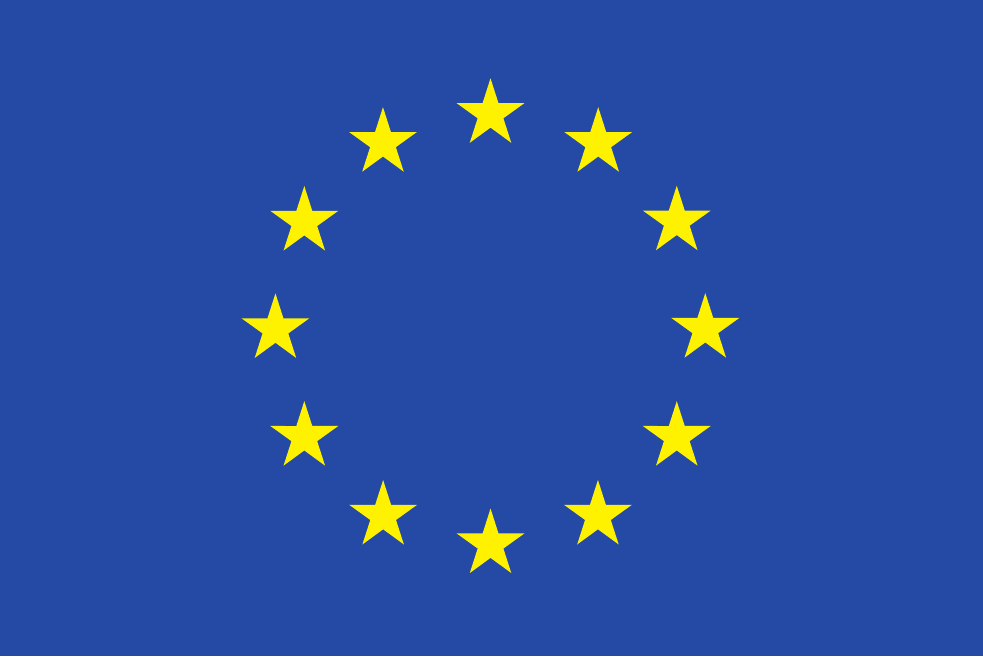}};
\node at (-0.2*\hd,0) {\includegraphics[height=1.2cm]{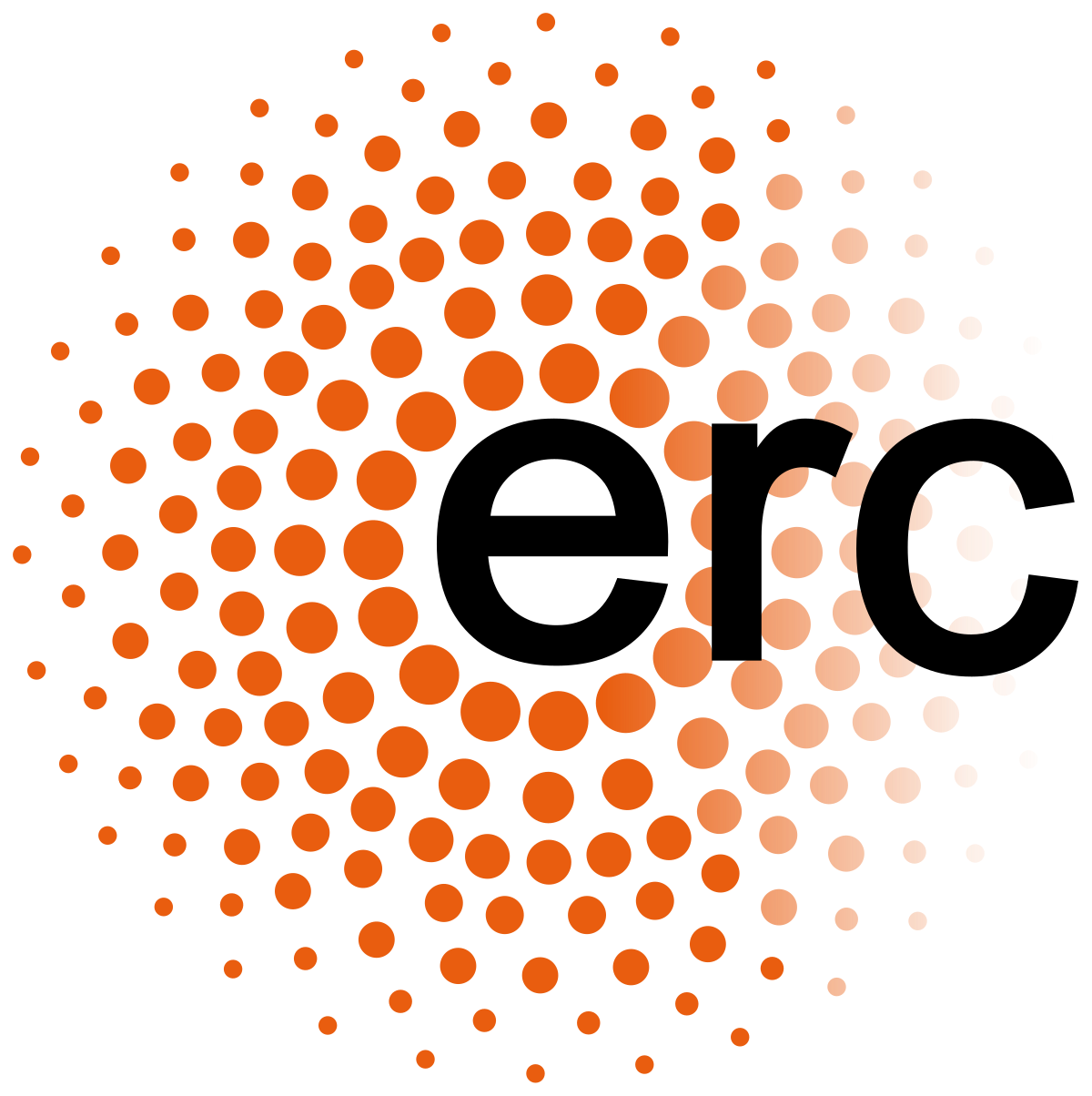}};
\end{scope}
\end{tikzpicture}

 \section{Introduction}

As more and more decisions are automated, there has been an increasing interest in incorporating fairness aspects in algorithms by design.
This applies in particular to clustering problems, where considerable attention has recently been dedicated to developing and studying various models of fair clustering, see, e.g.,~\cite{CKLV17},~\cite{BIPV19}, and~\cite{BCCN21}.

In this paper, we focus on the so-called colorful setting, which was introduced in~\cite{BIPV19}.
In colorful clustering, each client is a member of certain subgroups and every clustering is required to cover at least a given number of clients of each subgroup. This may be considered under various clustering objectives (like $k$-median and $k$-mean), though only the $k$-center case has been studied so far.

Colorful clustering is an appealing notion as it is a natural generalization of the robust (or outlier) setting, where there is only a single group which every client belongs to.
Various clustering problems have been studied in depth in the robust setting, see, e.g.,~\cite{CN19},~\cite{HPST19}, and~\cite{BCCN21}.

While the robust setting is amenable to a variety of well-known and basic algorithmic techniques, the only constant-factor approximations for the colorful setting, which imposes multiple covering constraints leading to more balanced clusterings, are based on significantly more sophisticated techniques, tailored specifically to those settings.
More precisely, three distinct techniques have been successful at achieving constant-factor approximations in the context of colorful $k$-center clustering, namely
the combinatorial approach of~\cite{JSS21},
the round-or-cut-based approach of~\cite{AAKZ21}, and
the iterative greedy reductions of~\cite{IV21}.

However, these approaches do not immediately generalize to variants with constraints on the facilities, even for the common Matroid Center or Knapsack Center clustering variants.
On the other hand, techniques for the Knapsack and Matroid $k$-Center problems in the robust setting (see~\cite{CN19} and~\cite{HPST19}) do not easily extend to multiple covering constraints.

Thus, prior to this work, no approaches have been known that lead to constant-factor approximations for colorful variants of otherwise well-studied $k$-center problems like Matroid Center or Knapsack Center.
Filling this gap is the goal of this paper.%
\subsection{Our contributions} 
%
Our main contribution is a partitioning procedure which leads to a general reduction of colorful $k$-center clustering problems with constraints on the facilities to a significantly simpler multi-dimensional covering problem (see \Cref{thm:mainRed}).
This reduction comes at the cost of a constant factor depending on the number of colors.

It is inspired by recent insights of~\cite{IV21} on decoupling multiple covering constraints and iteratively applying a greedy partitioning procedure of~\cite{CKMN01}.
By taking into account multiple colors at the same time, our framework gives an improved way of dealing with multiple covering constraints while also becoming more versatile.
Our framework also extends and simplifies ideas of the approximation algorithm for Robust Matroid Center of~\cite{CLLW16}.

We start by introducing the $\gamma$-Colorful $\mathcal{F}$-Supplier problem, which formalizes colorful $k$-center problems with (down-closed) constraints on the facilities.
\begin{definition}[$\gamma$-Colorful $\mathcal{F}$-Supplier problem]
Let $(C\discup F,d)$ be a finite metric space on a set of clients $C$ and facilities $F$, let $\mathcal F \subseteq 2^F$ be a down-closed family of subsets of $F$, and let $\gamma\in \mathbb{Z}_{\geq 0}$.
Moreover, we are given  for each $\ell\in [\gamma]$:
\begin{itemize}
    \item a unary encoded weight/color function $w_\ell: C \to \mathbb{Z}_{\geq 0}$, and
    \item a covering requirement $m_\ell\in \mathbb{Z}_{\geq 0}$.
\end{itemize}
The \emph{$\gamma$-Colorful $\mathcal{F}$-Supplier} problem asks to find the smallest radius $r$ together with a set $S\subseteq \mathcal{F}$ such that
$w_{\ell}(B_C(S,r))\geq m_{\ell}$ for all $\ell\in [\gamma]$.%
\footnote{We use the common notation $w(T)\coloneqq \sum_{t\in T} w(t)$ for functions $w\colon U\to \mathbb{R}_{\geq 0}$ and $T\subseteq U$, as well as $B(q,r)\coloneqq \{v \in C\cup F \mid d(q,v)\leq r\}$ for the ball of radius $r$ around point $q$. Moreover, we use the shorthand $B_U(V,r) \coloneqq \left\{ U \cap \bigcup_{v \in V} B(v,r) \right\}$ for sets $U,V\subseteq C\cup F$.}
\end{definition}
We note that it is also common to define colorful $k$-center versions in an unweighted way (thus not using weight functions $w_\ell$) by assigning to each client a subset of the $\gamma$ many colors and requiring that, for each color, $m_\ell$ many clients of that color are covered.
The definition we use clearly captures this case (and can easily be seen to be equivalent).
This connection also explains why the weights $w_\ell$ are assumed to be given in unary encoding.

Following common terminology in the literature, when $\mathcal{F}$ is the family of independent sets of a matroid or feasible sets with respect to a knapsack constraint, we call the problem \emph{$\gamma$-Colorful Matroid Supplier} and \emph{$\gamma$-Colorful Knapsack Supplier}, respectively.

Our main contribution is a general reduction of $\gamma$-Colorful $\mathcal{F}$-Supplier to an auxiliary problem, which we call \textsc{$\mathcal{F}$-Cover-Promise} (\AUXPROBLEM).
\AUXPROBLEM, which is formally defined below, is a multi-dimensional cover problem with the added promise that highly structured solutions exist.
The promise is key, as the problem without the promise can be thought of as a multi-dimensional max-cover problem.
\begin{definition}[\textsc{$\mathcal{F}$-Cover-Promise} (\AUXPROBLEM)]
\label{def:auxproblem}
In the {\normalfont\textsc{$\mathcal{F}$-Cover-Promise}} problem (\AUXPROBLEM), we are given a set family $\mathcal{H} \subseteq 2^{\mathcal{U}}$ over a finite universe $\mathcal{U}$,  a family $\mathcal{F}\subseteq 2^{\mathcal{H}}$ of feasible subsets of $\mathcal{H}$, and $\gamma$ many unary encoded weight functions $w_1,\ldots, w_\gamma: \mathcal{U}\to \mathbb{R}_{\geq 0}$ each with a requirement $m_\ell$ (for $\ell \in [\gamma]$).
The task is to find a feasible family of sets $S\in\mathcal{F}$ such that
\begin{equation*}
w_{\ell}\left(\bigcup_{H\in S} H\right) \geq m_\ell  \qquad \forall\; \ell\in [\gamma]\enspace.
\end{equation*}
The promise is that there exists a family $S\subseteq \mathcal{F}$ and a way to pick for each $H\in S$ a single representative $u_H\in H$ such that
\begin{equation*}
w_\ell \left(\{u_H \colon H\in S\}\right) \geq m_\ell \qquad \forall\; \ell\in [\gamma]\enspace.
\end{equation*}
\end{definition}
In words, the promise is that there is a solution that picks a family of sets and the requirements can be fulfilled by only using a single representative $u_H$ in each set.
However, the solution we are allowed to build is such that the weight of all elements covered by our sets are counted instead of just a single representative per set.

We are now ready to state our main reduction theorem, which, as we discuss later, readily leads, for a constant number of colors $\gamma$, to the first constant-factor approximations for $\gamma$-Colorful Matroid Supplier for linear matroids and $\gamma$-Colorful Knapsack Supplier.
Our reduction to \AUXPROBLEM comes at the cost of an $O(2^\gamma)$-factor in the approximation guarantee.
\begin{theorem}\label{thm:mainRed}
For any family of down-closed set systems, we have that if \AUXPROBLEM can be solved efficiently for any $\mathcal{F}$ in that family, then there is an $O(2^\gamma)$-approximation algorithm for $\gamma$-Colorful $\mathcal F$-Supplier for any $\mathcal{F}$ in the family.%
\footnote{When talking about the same set system $\mathcal{F}$ both in the context of \AUXPROBLEM and $\gamma$-Colorful $\mathcal{F}$-Supplier, we consider $\mathcal{F}$ to be the same set system in both settings even if the ground sets are different, as long as there is a one-to-one relation between the ground sets mapping sets of one system to sets of the other one and vice versa.
}
\end{theorem}

While the dependence of the approximation factor on $\gamma$ may be undesirable, the algorithmic barriers for prior approaches remain even when $\gamma =2$ and, for hardness reasons, we do not expect approximation algorithms to exist at all when $\gamma$ grows too quickly.
In particular,~\cite{AAKZ21} showed that even a simple version of colorful clustering, where any $k$ centers can be chosen, does not admit an $O(1)$-approximation algorithm when $\gamma = \omega (\log |C \discup F|)$ under the Exponential Time Hypothesis.
Thus, in what follows, we restrict ourselves to $\gamma=O(1)$.

\medskip

We now discuss implications of~\Cref{thm:mainRed} to $\gamma$-Colorful Matroid Supplier for linear matroids and $\gamma$-Colorful Knapsack Supplier.
When $\mathcal F$ is the family of independent sets of a linear matroid, we show how \AUXPROBLEM can be solved with techniques relying on an efficient randomized procedure for the \emph{Exact Weight Basis} (XWB) problem for linear matroids.%
\footnote{In XWB, one is given a matroid on a ground set with unary encoded weights and a target weight; the goal is to find a basis of the matroid of weight equal to the target weight.
The technique in~\cite{CGM92} to solve XWB for linear matroids needs an explicit linear representation of the linear matroid.
We make the common assumption that this is the case whenever we make a statement about linear matroids.
} 
Linear matroids include as special cases many other well-known matroid classes, including uniform matroids, and more generally partition and laminar matroids, graphic matroids, transversal matroids, gammoids, and regular matroids.
\begin{theorem}
\label{thm:matroid}
For $\gamma=O(1)$ and $\mathcal{F}$ being the independent sets of a linear matroid, $\AUXPROBLEM$ can be solved efficiently by a randomized algorithm.
Hence (by~\Cref{thm:mainRed}), there is a randomized $O(2^\gamma)$-approximation algorithm for $\gamma$-Colorful Matroid Supplier for linear matroids.
\end{theorem}

The restriction to linear matroids and the fact that the algorithm is randomized are not artifacts of our framework.
Indeed, by an observation in~\cite{JSS21}, rephrased for matroids below, we do not only have that XWB implies results for $\gamma$-Colorful Matroid Supplier (which will follow from our reduction), but also a reverse implication.
More precisely, even for $2$-Colorful Matroid Supplier, deciding whether there is a solution of radius zero requires being able to solve XWB on that matroid.
However, it is unknown whether XWB can be solved efficiently on general matroids, and the only technique known for XWB on linear matroids is inherently randomized~\cite{CGM92}. (Derandomization is a long-standing open question in this context.)
\begin{lemma}[based on \cite{JSS21}]
\label{lem:hardness}
If there is an efficient algorithm for deciding whether $2$-Colorful Matroid Supplier with respect to a given class of matroids admits a solution of radius zero, then XWB can be solved efficiently on the same class of matroids.
\end{lemma}
Note that if we cannot decide the existence of a radius zero solution, then no approximation algorithm with any finite approximation guarantee can exist.

\medskip

For the case where $\mathcal{F}$ are the feasible sets for a knapsack problem, one can use standard dynamic programming techniques to see that \AUXPROBLEM can be solved efficiently, which readily leads to a $O(2^\gamma)$-approximation for $\gamma$-Colorful Knapsack Supplier.

\medskip

Whereas our reduction given by \Cref{thm:mainRed} is broadly applicable and readily leads to first constant-factor approximations for $\gamma$-Colorful $\mathcal{F}$-Supplier problems, it remains open whether and in which settings a dependence of the approximation factor on the number of colors is necessary.
We make first progress toward this question for $\gamma$-Colorful Knapsack Supplier, where we show how techniques from \cite{JSS21} can be modified and extended to give a $7$-approximation (independent of the number of colors). 
\begin{theorem}
\label{thm:knapsack}
For $\gamma = O(1)$, there is a $7$-approximation algorithm for $\gamma$-Colorful Knapsack Supplier.
\end{theorem}

Our technical contribution here lies in handling the knapsack constraint in this approach--- modifying the algorithm of \cite{JSS21} to the supplier setting and to weighted instances is straight-forward.
In fact, their algorithm can be seen to give a $3$-approximation even for $\gamma$-Colorful $k$-Supplier%
, which is tight in light of a hardness result in~\cite{CKMN01}, namely that it is \NP-hard to approximate Robust $k$-center with forbidden centers to within $3-\epsilon$.
This remains the strongest hardness result even for $\gamma$-Colorful $\mathcal F$-Supplier problems.

\subsection{Organization of this paper}

Our main reduction, \Cref{thm:mainRed}, is based on what we call $(L,r)$-partitions, which is a way to judiciously partition the clients into parts that we want to cover together.
We introduce $(L,r)$-partitions in \Cref{sec:redTroughLDecomp} and show how the existence of certain strong $(L,r)$-partitions implies \Cref{thm:mainRed}.
In \Cref{sec:applicationsOfRed}, we show how our reduction framework can be used to obtain first constant-factor approximations for $\gamma$-Colorful Matroid Supplier for linear matroids (thus showing \Cref{thm:matroid}) and $\gamma$-Colorful Knapsack Supplier.
Finally, in \Cref{sec:LDecompProof} we prove existence of strong $(L,r)$-partitions.
The proof of \Cref{lem:hardness} and our $7$-approximation for $\gamma$-Colorful Knapsack Supplier, i.e., the proof of \Cref{thm:knapsack}, are presented in \Cref{app:hardness} and \Cref{app:knapsack}, respectively.

\section{Reducing to \texorpdfstring{$\mathcal{F}$}{F}-CP through \texorpdfstring{$(L,r)$}{(L,r)}-partitions}
\label{sec:redTroughLDecomp}

Consider a $\gamma$-Colorful $\mathcal{F}$-Supplier problem on a metric space $(X=(C\discup F),d)$ with weights $w_\ell\colon C\to \mathbb{Z}_{\geq 0}$ for $\ell\in [\gamma]$ and covering requirements $m_\ell\in \mathbb{Z}_{\geq 0}$ for $\ell\in [\gamma$].
An $(L,r)$-partition is a partition of the clients into parts of small diameter each of which we consider in our analysis to be either fully covered or not covered at all.
The key property of an $(L,r)$-partition is that, if our instance admits a radius-$r$ solution, then there is a radius-$(L+1)r$ solution where we allow each center to cover only \emph{a single part} of the partition.
It is the existence of such highly structured solutions that we exploit to design $O(1)$-approximation algorithms.

A crucial property of $(L,r)$-partitions is that they neither depend on $\mathcal{F}$ nor the covering requirements $m_\ell$, but only on the metric space and the weight functions, which we call a $\gamma$-colorful space for convenience.
\begin{definition}[$\gamma$-colorful space $(X,d,w)$]
A $\gamma$-colorful space $(X=C\discup F,d,w)$ consists of 
\begin{enumerate}
\item a metric space $(X, d)$, and
\item color functions $w_\ell\colon C \to \mathbb{R}_{\geq 0}$ for $\ell\in [\gamma]$.
\end{enumerate}
\end{definition}

We assume for convenience that the supports of the color functions, i.e., $\supp(w_\ell)$ for $\ell \in [\gamma]$, are pairwise disjoint.
One can reduce to this case without loss of generality by co-locating copies of clients.
We are now ready to formally define the notion of $(L,r)$-partition.%
\begin{definition}[$(L,r)$-partition]
\label{def:L_r_partition}
Let $(X=C\discup F,d,w)$ be a $\gamma$-colorful space and $r,L\in\mathbb{R}_{\geq 0}$.
A partition $\mathcal{P}\subseteq 2^C$ is an \emph{$(L,r)$-partition} if
\smallskip
\begin{enumerate}
\item\label{item:LrPart_dist} $\diam(A)\coloneqq \max_{u,v \in A} d(u,v) \leq L\cdot r \quad \forall A \in \mathcal{P}$, and
\smallskip
\item \label{item:LrPart_z} for any $Z\subseteq F$, there exists a subfamily $\mathcal{A}\subseteq \mathcal{P}$ and injection $h:\mathcal{A}\to Z$ such that
\smallskip
\begin{enumerate}
\item\label{item:LrPart_shortAssign} $d(A,h(A))\leq r$,\footnote{
For any set $V\subseteq F\discup C$ and $x\in F\discup C$, we use the shorthand $d(V,x)\coloneqq \min\{d(v,x)\colon v\in V\}$.
}
and
\item\label{item:LrPart_coverage} $w_{\ell}\left(\bigcup_{A\in \mathcal{A}} A\right) \geq w_\ell\left(B_C(Z,r)\right) \quad \forall \ell\in [\gamma]$.
\end{enumerate}
\end{enumerate}
\end{definition}

\begin{figure}
    \centering
    \includegraphics[width=\textwidth]{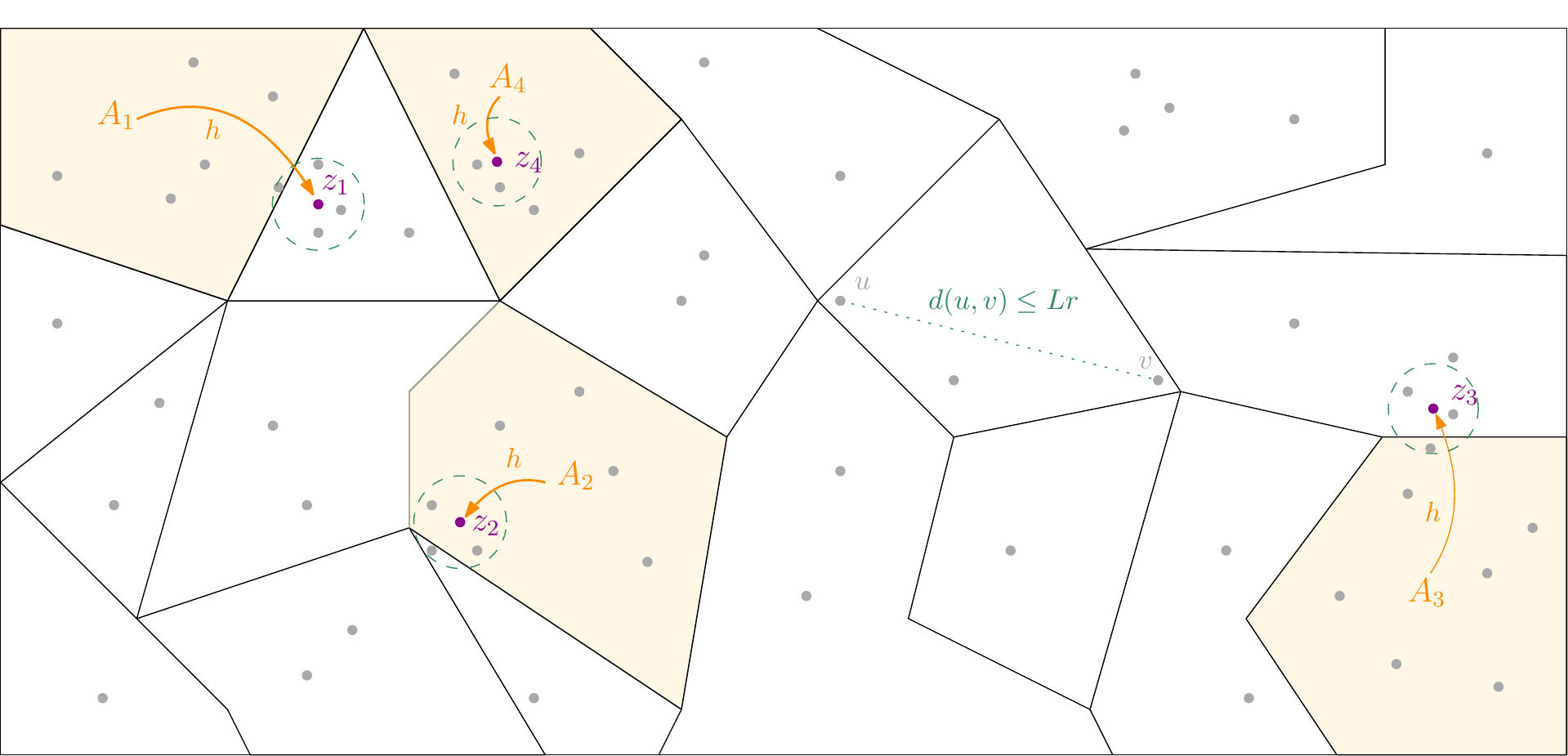}
    \caption{Illustration of an $(L,r)$-partition of a $1$-colorful space (where all points have unit weight). For $Z=\{ z_i \mid i \in [4]\}$, the mapping $h$ maps $A_i$ to $z_i$ for $i \in [4]$. Note that $\cup_{i \in [4]}A_i$ contains at least as many points than $\cup_{i \in [4]} B(r,z_i)$ 
    and that $d(z_i,A_i) \leq r$ for $i \in [4]$.
    Furthermore, the largest distance between any two points in a set $A_i$ is bounded by $Lr$.
    }
    \label{fig:partition_definition}
\end{figure}

To connect $(L,r)$-partitions to colorful clustering problems, think of $Z\in \mathcal{F}$ as centers of a $\gamma$-Colorful $\mathcal{F}$-Supplier problem that satisfy the covering requirements with radius $r$.
The definition of an $(L,r)$-partition $\mathcal{P}$ then implies that there is a subset $\mathcal{A} \subseteq \mathcal{P}$ of the parts such that
(i) for each $A\in \mathcal{A}$ there exists an element $h(A) \in Z$ such that any client in $A$ has distance at most $(L+1)\cdot r$ from $h(A)$, which follows from property~\ref{item:LrPart_dist} and \ref{item:LrPart_shortAssign} of the definition, and
(ii) the clients in $\mathcal{A}$ cover as much as $B_C(Z,r)$ in each color.
Thus, the set of facilities $h(\mathcal{A})$ satisfies the covering requirements with respect to the radius $(L+1)\cdot r$, and, furthermore, $h(\mathcal{A})$ is feasible because $h(\mathcal{A})\subseteq Z$ and $\mathcal{F}$ is down-closed. In short, $h(\mathcal{A})$ is an $(L+1)$-approximate solution to the $\gamma$-Colorful $\mathcal{F}$-Supplier problem.
Hence, to obtain an $(L+1)$-approximation, the problem reduces to deciding which of the parts of $\mathcal{P}$ to cover.
A key simplification we gain from this connection is that the client sets in $\mathcal{P}$ are non-overlapping because $\mathcal{P}$ is a partition, which we will heavily exploit later to design our algorithms.

The key structural result of our work is to show that $(L,r)$-partitions with constant $L$ (for a fixed $\gamma$) exist and can also be constructed efficiently, which is summarized below.

\begin{lemma}\label{lem:main_part}
For every $\gamma$-colorful space $(X,d,w)$ and $r\in \mathbb{R}_{\geq 0}$, one can construct in polynomial time a $(10(2^\gamma-1),r)$-partition.\footnote{As we highlight later, a more careful analysis of our approach allows for a slight improvement in the constant factor, leading to the construction of $(8\cdot 2^\gamma - 10,r)$-partitions.
However, in the interest of simplicity, we present a simpler analysis that shows the bound claimed in the lemma.}
\end{lemma}

We defer the proof of \cref{lem:main_part} to \cref{sec:LDecompProof}, and first show how it implies our main reduction theorem, \cref{thm:mainRed}, and how this reduction readily leads to $O(1)$-approximations for $\gamma$-Colorful Matroid Supplier for linear matroids and $\gamma$-Colorful Knapsack Supplier.

\begin{proof}[Proof of \cref{thm:mainRed}]
Consider an instance of $\gamma$-Colorful $\mathcal F$-Supplier on a $\gamma$-colorful space $(X,d,w)$.
We can guess the radius $r$ of an optimal solution to the problem.
This can be achieved by considering all pairwise distances between facilities $F$ and clients $C$, repeating the steps below for each guess and only considering the best output (and discarding outputs where the procedure fails).
Hence, assume that $r$ is the optimal radius from now on.

By \cref{lem:main_part}, we can efficiently construct an $(L,r)$-partition $\mathcal{P}$ of $(X,d,w)$ for $L=10(2^\gamma-1)=O(2^{\gamma})$.
Consider the \AUXPROBLEM instance with universe $\mathcal{U} \coloneqq  \mathcal{P}$,
family of sets
\begin{align*}
\mathcal{H} &\coloneqq \left\{H_f\colon f\in F\right\}\enspace,\text{where}\\
        H_f &\coloneqq \left\{ A\in \mathcal{P} \text{ with } d(A,f)\leq r \right\} \quad \forall\; f\in F \enspace .
\end{align*}
The family of feasible subsets of $\mathcal{H}$ is the same as $\mathcal{F}$ when identifying $H_f$ with the element $f$. 
To make this relation explicit, if we denote by $\mathcal{F}_{\mathcal{H}}$ the family of feasible subsets, then some subset of $\mathcal{H}$, say $\{H_f \colon f\in I\}$ where $I\subseteq F$, is in $\mathcal{F}_{\mathcal{H}}$ if and only if $I\in \mathcal{F}$.
Moreover, the weights and coverage thresholds are inherited from those of the given $\gamma$-Colorful $\mathcal{F}$-Supplier problem; formally, for $\ell\in [\gamma]$, the $\ell$-th weight of $A\in \mathcal{U}$ is given by $w_\ell(A)$.

To make sure that this indeed leads to an $\mathcal{F}$-CP problem, we have to verify that the promise holds.
Thus, let $Z\subseteq F$ be a solution to the given $\gamma$-Colorful $\mathcal{F}$-Supplier problem for radius $r$, which exists because we assume that $r$ was guessed correctly.
As $\mathcal{P}$ is an $(L,r)$-partition of $(X,d,w)$, there is a subfamily $\mathcal{A} \subseteq \mathcal{P}$ and injection $h\colon\mathcal{A} \to Z$ satisfying property~\ref{item:LrPart_z} of \cref{def:L_r_partition}.
We claim that a solution fulfilling the promise is given by choosing
\begin{equation*}
S = \left\{H_f \colon f\in h(\mathcal{A})\right\} \in \mathcal{F}_{\mathcal{H}}\enspace,
\end{equation*}
and setting as representative element $u_{H_f}\in H_f$ the element $u_{H_f}=A_f$, where 
$A_f=h^{-1}(f)$.
Note that because $h(\mathcal{A})\subseteq Z\in \mathcal{F}$ and $\mathcal{F}$ is down-closed, we indeed have $S\in \mathcal{F}_{\mathcal{H}}$.
Furthermore, because the injection $h$ satisfies $d(A_f, h(A_f))\leq r$, we have $u_{H_f}\in H_f$, as desired.
Moreover, 
\begin{align*}
\sum_{f\in h(\mathcal{A})}w_\ell\left(u_{H_f}\right) = \sum_{f\in h(\mathcal{A})} w_\ell\left(A_f\right) = \sum_{A\in \mathcal{A}} w_\ell(A) \geq w_\ell(B_C(Z,r)) \geq m_\ell \quad \forall \ell\in [\gamma]\enspace,
\end{align*}
where the first inequality follows because $\mathcal{A}$ fulfills the second property of \cref{def:L_r_partition}, and the last inequality is a consequence of $Z$ being centers that are a radius-$r$ solution to the given $\gamma$-Colorful $\mathcal{F}$-Supplier problem.
Hence, the promised solution exists.

Thus, we can compute an \AUXPROBLEM solution $S_{\mathcal{H}}\subseteq\mathcal{F}_{\mathcal{H}}$, which can be written as $S_{\mathcal{H}}\coloneqq \{H_f\colon f\in S\}$ for some $S\in \mathcal{F}$.
We claim that $S$ is a solution to the given $\gamma$-Colorful $\mathcal{F}$-Supplier problem with radius $(L+1)\cdot r$, which finishes the proof.
This follows from the fact that $S_{\mathcal{H}}$ is an \AUXPROBLEM solution, and that, for any $f\in F$, each client in $\bigcup_{A\in H_f} A$ has distance at most $(L+1)\cdot r$ from $f$ because $\mathcal{P}$ is an $(L,r)$-partition.
Hence, the clustering solution with centers $S$ and radius $(L+1)\cdot r$ covers all clients in
\begin{equation*}
\bigcup_{f\in S} \bigcup_{A \in H_f} A\enspace,
\end{equation*}
and the $w_\ell$-weight (for any $\ell\in [\gamma]$) that it covers is at least 
\begin{equation*}
w_\ell\left(\bigcup_{f\in S} \bigcup_{A\in H_f} A \right)
= \sum_{A \in \bigcup_{f \in S} H_f} w_\ell(A) \geq m_\ell
\enspace,
\end{equation*}
where the equality uses that the ground set $\mathcal{U}=  \mathcal{P}$ consists of sets $A$ that are disjoint, and the inequality holds because $S_{\mathcal{H}}=\{H_f\colon f\in S\}$ is a solution to $\AUXPROBLEM$.
Thus, all coverage requirements are fulfilled by the clustering with centers $S$ and radius $(L+1)\cdot r$, as desired.
\end{proof}

\section{Applications of our reduction framework}\label{sec:applicationsOfRed}

We now discuss implications of our reduction framework, \cref{thm:mainRed}, to $\gamma$-Colorful Matroid Supplier for linear matroids and $\gamma$-Colorful Knapsack Supplier.

\subsection{\texorpdfstring{$\gamma$}{gamma}-Colorful Matroid Supplier}
\label{subsec:matroid}

To apply our reduction framework to $\gamma$-Colorful Matroid Supplier for linear matroids, we have to solve \AUXPROBLEM when $\mathcal{F}$ are the independent sets of a linear matroid.
We show how this problem can be reduced to XWB in a suitably defined matroid.
More precisely, we use a reduction to the \emph{Exact Weight Independent Set} (XWI) problem for matroids.
This problem is identical to XWB except that an \emph{independent set} with the desired target weight needs to be returned, instead of a basis.
However, XWI easily reduces to XWB on linear matroids, by adding zero weight copies of the elements.

This reduction relies on Rado matroids, which is a way to construct a matroid from another one (see, e.g., \cite[Section~8.2]{welsh2010matroid}).%
\footnote{This construction of Rado matroids is also called the \emph{induction of a matroid by a bipartite graph}.}
It relies on the notation of a \emph{system of representatives}, where, for a finite universe $\mathcal{U}$ and a set system $S\subseteq 2^{\mathcal{U}}$, a \emph{system of representatives of $S$} is any set $\{u_H\}_{H\in S}$ with $u_H\in H$ for $H\in S$.
In words, a system of representatives is obtained by replacing each set in $S$ by an element in that set (its representative). (Note that an element can be chosen more than once as a representative, but, as defined above, only appears once in the system of representatives.)

\begin{definition}[Rado matroid]
\label{def:induced}
Let $\mathcal U$ be a finite universe, $\mathcal H \subseteq 2^\mathcal U$ be some set system, and let $M=(\mathcal H, \mathcal I)$ be a matroid.
The Rado matroid $(\mathcal{U},\overline{\mathcal{I}})$ induced by $(\mathcal U, \mathcal H, M)$ is a matroid on the ground set $\mathcal{U}$ with independent sets
\begin{equation*}
\{U\subseteq \mathcal{U} \colon U \text{ is a system of representatives for some $I\in \mathcal{I}$}\}\enspace.
\end{equation*}
\end{definition}
A proof that a Rado matroid is indeed a matroid can be found, e.g., in \cite[Section~8.2]{welsh2010matroid}. 
We will reduce \AUXPROBLEM to XWI on a Rado matroid obtained from a linear matroid.
For this, we need that also the Rado matroid we obtain is linear and, moreover, that an explicit linear representation of it can be found efficiently, which is the case due to a result from~\cite{PW70}.
\begin{lemma}[see Theorem~3 of \cite{PW70}]
\label{lem:linear_rado_matroid}
For a set family $\mathcal{H} \subseteq 2^{\mathcal{U}}$ and a linear matroid $M=(\mathcal{H}, \mathcal{I})$, the Rado matroid $\overline{M}=(\mathcal{U}, \overline{\mathcal{I}})$ induced by $(\mathcal U, \mathcal H, M)$ is a linear matroid.
Moreover, given a linear representation of $M$, one can find a linear representation of $\overline{M}$ in time polynomial in $|\mathcal{H}|$, $|\mathcal{U}|$, and the size of the linear representation of $M$.
\end{lemma}

We are now ready to show that \AUXPROBLEM can be solved efficiently for linear matroids, which implies \cref{thm:matroid}.
\begin{lemma}\label{lem:aux_matroid}
\AUXPROBLEM can be solved efficiently when $\mathcal{F}$ is the family of independent sets of a linear matroid.
\end{lemma}
\begin{proof}
We recall that we are given an \AUXPROBLEM instance, which defines a set system $\mathcal{H}\subseteq 2^{\mathcal{U}}$ over a finite universe $\mathcal{U}$, and a family $\mathcal{F}\subseteq 2^{\mathcal{H}}$ such that $M=(\mathcal{H},\mathcal{F})$ is a linear matroid.
Let $\overline{M}=(\mathcal{U},\overline{\mathcal{I}})$ be the Rado matroid induced by $(\mathcal U, \mathcal H, M)$.
$\overline{M}$ is a linear matroid by \cref{lem:linear_rado_matroid} and we can obtain a linear representation of $\overline{M}$ in polynomial time.
The promise of \AUXPROBLEM implies the existence of an independent set $T$ of $\overline{M}$ satisfying the covering requirements, i.e.,
\begin{equation}\label{eq:promiseMat}
w_{\ell}(T) \geq m_{\ell} \qquad \forall \ell \in [\gamma]\enspace.
\end{equation}

To solve \AUXPROBLEM, we guess, for each color $\ell\in [\gamma]$, the weight $\lambda_\ell\coloneqq w_\ell(T)$ that $T$ covers.
Note that $\lambda_\ell$ is at most $W_\ell\coloneqq w_\ell(\mathcal{U})$, which, due to the unary encoding of $w_\ell$, is polynomially bounded in the input.
Hence, the guessing of the $\lambda_\ell$, for $\ell\in [\gamma]$, can be performed in time $\prod_{\ell\in [\gamma]} W_\ell$, which is polynomially bounded because $\gamma=O(1)$.

We now determine an independent set $\widetilde{T}$ in $\overline{M}$ with $w_\ell(\widetilde{T})=w_\ell(T)$ for each $\ell\in [\gamma]$.
This can be achieved by encoding all $\ell$ many (unary encoded) weight functions $w_\ell$ for $\ell \in [\gamma]$ into a single one $\overline{w}$ and then solving an appropriate XWI problem with respect to $\overline{w}$.
More precisely, for an element $u\in \mathcal{U}$, we obtain a new single weight $\overline{w}(u)$ whose first $\lceil\log_2(|W_1|+1)\rceil$ bits represent the weight $w_1(u)$, the next $\lceil\log_2(|W_2|+1)\rceil$ bits the weight $w_2(u)$, and so on.
Because $\gamma=O(1)$ and all $w_\ell$ have unary encoding, this leads to combined weights $\overline{w}$ whose unary encoding is polynomially bounded.
Analogously, we encode the guessed weights $\lambda_\ell$ for $\ell\in [\gamma]$ into a single one $\overline{\lambda}$.
We now solve XWI on $\overline{M}$ with weights $\overline{w}$ and target weight $\overline{\lambda}$.
As $\overline{M}$ is linear, this is possible by a randomized algorithm in time pseudo-polynomial in the total weight~\cite{CGM92}.
Moreover, because the weights are unary encoded in our setting, this implies a polynomial running time as desired.

Let $\widetilde{T}$ be a solution of this XWI problem, which must exist for the correct guess of the $\lambda_\ell$ because of the promised solution $T$.
$\widetilde{T}$ being independent in $\overline{M}$ implies that it is a system of representatives for some independent set $S\in \mathcal{F}$ of $M$.
Such a set $S$ can be found through matroid intersection.
More precisely, it is known that the minimal (inclusion-wise) sets $I \subseteq \mathcal{H}$ such that $\widetilde{T}$ is a system of representatives for $I$ form the basis of a matroid $\widetilde{M}$, for which an efficient independence oracle can be obtained. (See \cite[Section~7.3]{welsh2010matroid}.)
Hence, the desired set $S$ can be obtained by finding a basis of $\widetilde{M}$ that is independent in $M$, which can be computed through matroid intersection algorithms.
The set $S$ is the solution of \AUXPROBLEM that we return.
Because $\widetilde{T}\subseteq \bigcup_{H\in S} H$, the set $S$ fulfills the covering requirements due to~\eqref{eq:promiseMat}.
\end{proof}%
\subsection{\texorpdfstring{$\gamma$}{gamma}-Colorful Knapsack Supplier}
\label{subsec:knapsack}

To showcase the versatility of our reduction, we now show how it implies an $O(2^\gamma)$-approximation for $\gamma$-Colorful Knapsack Supplier, by discussing an efficient way to solve \AUXPROBLEM when $\mathcal{F}$ are the feasible solutions to a knapsack constraint.
Even though there is a stronger (and more sophisticated) approximation result for this problem (as stated in \Cref{thm:knapsack}), this application is a nice example of how one can readily obtain constant-factor approximations through our reduction technique combined with known methods; in this case, by solving \AUXPROBLEM through a standard dynamic programming approach.

\begin{lemma}
Let $\mathcal F$ be the feasible sets of a knapsack constraint, i.e., $\mathcal F = \{S\subseteq \mathcal{H} \colon \kappa(S)\leq K\}$ for some $\kappa\colon \mathcal{H} \to \mathbb{R}_{\geq 0}$ and budget $K\in \mathbb{R}_{\geq 0}$. 
Then \AUXPROBLEM can be solved efficiently.
\end{lemma}
\begin{proof}
Recall that the \AUXPROBLEM problem to be solved defines a family $\mathcal{H}\subseteq 2^{\mathcal{U}}$ over a finite universe $\mathcal{U}$, and a family $\mathcal{F}\subseteq 2^{\mathcal{H}}$, which is defined by a knapsack constraint, i.e., $\mathcal{F}=\{S\subseteq \mathcal{H} \colon \kappa(S) \leq K \}$.
We define the following weight function on $\mathcal{U}$:
\begin{equation*}
\eta(u) \coloneqq \min\{\kappa(H) \colon H\in \mathcal{H} \text{ with } u\in H\}\enspace.
\end{equation*}
In words, $\eta(u)$ corresponds to the cost of the cheapest set in $\mathcal{H}$ that covers $u$.
Consider the following binary program, which can be solved efficiently by standard dynamic programming techniques due to the unary encoding of the weights $w_\ell$ for $\ell\in [\gamma]$ (see, e.g.,~\cite{AAKZ21} for details):
\begin{equation*}
\begingroup
\renewcommand*{\arraystretch}{1.5}
\begin{array}{r>{\displaystyle}rll@{\quad}l}
\min & \sum_{u\in \mathcal{U}} \eta (u) \cdot z(u) &                                           \\
     & \sum_{u\in \mathcal{U}} w_\ell(u)  \cdot z(u) & \geq  & m_\ell & \forall \ell\in [\gamma] \\
     & z                                             & \in   & \{0,1\}^{\mathcal{U}} \enspace.
\end{array}
\endgroup
\end{equation*}

We compute an optimal solution $z^*$ to the above binary program.
Let $Q \coloneqq \{u\in \mathcal{U} \colon z^*(u) =1\}$.
For each $u\in Q$, let $H_u\in \mathcal{H}$ be a set of minimum cost that contains $u$; hence, $\kappa(H_u)=\eta(u)$.
We claim that $\{H_u \colon u\in Q\}$ is a solution to \AUXPROBLEM.
Because $z^*$ fulfills the constraints of the binary program, we have that $\{H_u \colon u\in Q\}$ fulfills the covering requirements.
It remains to show that it fulfills the knapsack constraint, i.e., its cost is at most $K$.
This reduces to show that the optimal value of the binary program is at most $K$.
We claim that this holds because of the promise of \AUXPROBLEM.
Indeed, the promise guarantees that there is $S\subseteq \mathcal{F}$ and a system of representatives $u_H$ for $H\in S$ such that
$w_\ell(\{u_H \colon H\in S\})\geq m_\ell$ for $\ell\in [\gamma]$.
Hence, setting $z_{u_H}=1$ for all $H\in S$, and setting all other coordinates of $z\in \{0,1\}^{\mathcal{U}}$ to zero, is a solution to the binary program which has objective value at most $\kappa(S)\leq K$.
\end{proof}

\section{Existence and construction of strong \texorpdfstring{$(L,r)$}{(L,r)}-partitions}\label{sec:LDecompProof}
We now prove our key structural result, \cref{lem:main_part}, which guarantees the existence and efficient constructability of $(O(2^\gamma),r)$-partitions for $\gamma$-colorful spaces.
Our proof proceeds by induction on $\gamma$.
The base case, i.e., $\gamma=0$, holds because the family $\{\{c\} \colon c\in C\}$ is a $(0,r)$-grouping on every $0$-colorful space $\left(C\discup F, d, w\right)$.
The key step is extending an $(L,r)$-partition of a $(\gamma-1)$-colorful space to a suitable partition of a $\gamma$-colorful space.

To this end, we extend ideas on the greedy algorithm of~\cite{CKMN01}, which was originally introduced to deal with a single color $k$-center problem.
More precisely, to augment a partition of a $(\gamma -1)$-colorful space, we apply a greedy subroutine on the points of color $\gamma$.
A careful construction and analysis (which takes into account the earlier colors) then shows that this yields a $(2L+10,r)$-partition of the $\gamma$-colorful space. 
Our refined charging scheme improves on a decoupled analysis of \cite{IV21} (which gives an $O(5^{\gamma})$ approximation algorithm for $\gamma$-Colorful $k$-Center).

The lemma below formalizes the induction step.%
\begin{lemma}\label{lemma:part_inductionstep}
Given a $(L,r)$-partition for a $(\gamma-1)$-colorful space, then one can efficiently construct a $(2L+10,r)$-partition for any $\gamma$-colorful space obtained by adding one color to the $(\gamma-1)$-colorful space.
\end{lemma}
\begin{proof}
Let $\left(C\discup F, d, w\right)$ be a $\gamma$-colorful space, and let $\widehat{w}=(w_1,\ldots,w_{\gamma-1})$ be the first $\gamma-1$ colors. (Hence, we omitted the last color.)
Let $C_{\gamma}\coloneqq \supp(w_\gamma)$ and $C_{<\gamma}\coloneqq C\setminus C_{\gamma}$, and let $\mathcal{P}$ be a $(L,r)$-partition of the $(\gamma-1)$-colorful space $\left(C_{< \gamma}\discup F, d, \widehat{w}\right)$. 
Note that we assumed that the supports of the weights $w_\ell$ are disjoint.
Hence, $w_\ell(C_\gamma)=0$ for $\ell\in [\gamma-1]$.
Moreover, without loss of generality, we assume that for every client $c\in C$, there is a facility $f\in F$ with $d(f,c)\leq r$.
All clients not fulfilling this condition can be deleted from the instance without changing the statement as they can never be covered by any radius-$r$ solution.
Indeed, a partition of the clients of this purged instance can simply be extended to a partition of all clients by adding the deleted clients as singleton sets to the partition.

We now prove that \cref{algo:iterative_decomposition} returns an $(\overline{L},r)$-partition $\overline{\mathcal{P}}$ of $\left(C\discup F, d, w\right)$, where $\overline{L}\coloneqq 2L+10$.
\cref{algo:iterative_decomposition} goes through all facilities in a well-chosen order and iteratively builds new parts consisting of parts in $\mathcal{P}$ together with a subset of $C_\gamma$. (See \cref{fig:main_lemma} for an illustration of this procedure.)

\begin{algorithm}[ht]
\setstretch{1.3}
\caption{\textsc{GreedyPartitioning}$(C, F, d, \mathcal{P},w_{\gamma})$}
\label{algo:iterative_decomposition}
\For{$i=1$ \KwTo $|F|$}
{
$g_i = \underset{f \in F \backslash \{ g_1 , \ldots , g_{i-1}\}}{\argmax } \, w_{\gamma}\left (B_C(f, r) \setminus \bigcup_{t = 1}^{i-1} \overline{A}_t\right )$;\\
 $\displaystyle\overline{A}_i \leftarrow  \left(B_{C_\gamma}(g_i,3r) \cup \bigcup_{\substack{A \in \mathcal{P} \text{ with}\\ d(g_i,A)\leq 5r}} A \right)\ \setminus \ \bigcup_{t = 1 }^{i-1} \overline{A}_t $;\\
}
\Return $\overline{\mathcal{P}}\coloneqq\left\{\overline{A}_i\colon i\in \left [|F|\right]\right\}$;\\
\end{algorithm}
\begin{figure}[ht]
    \centering
    \includegraphics[width=\textwidth]{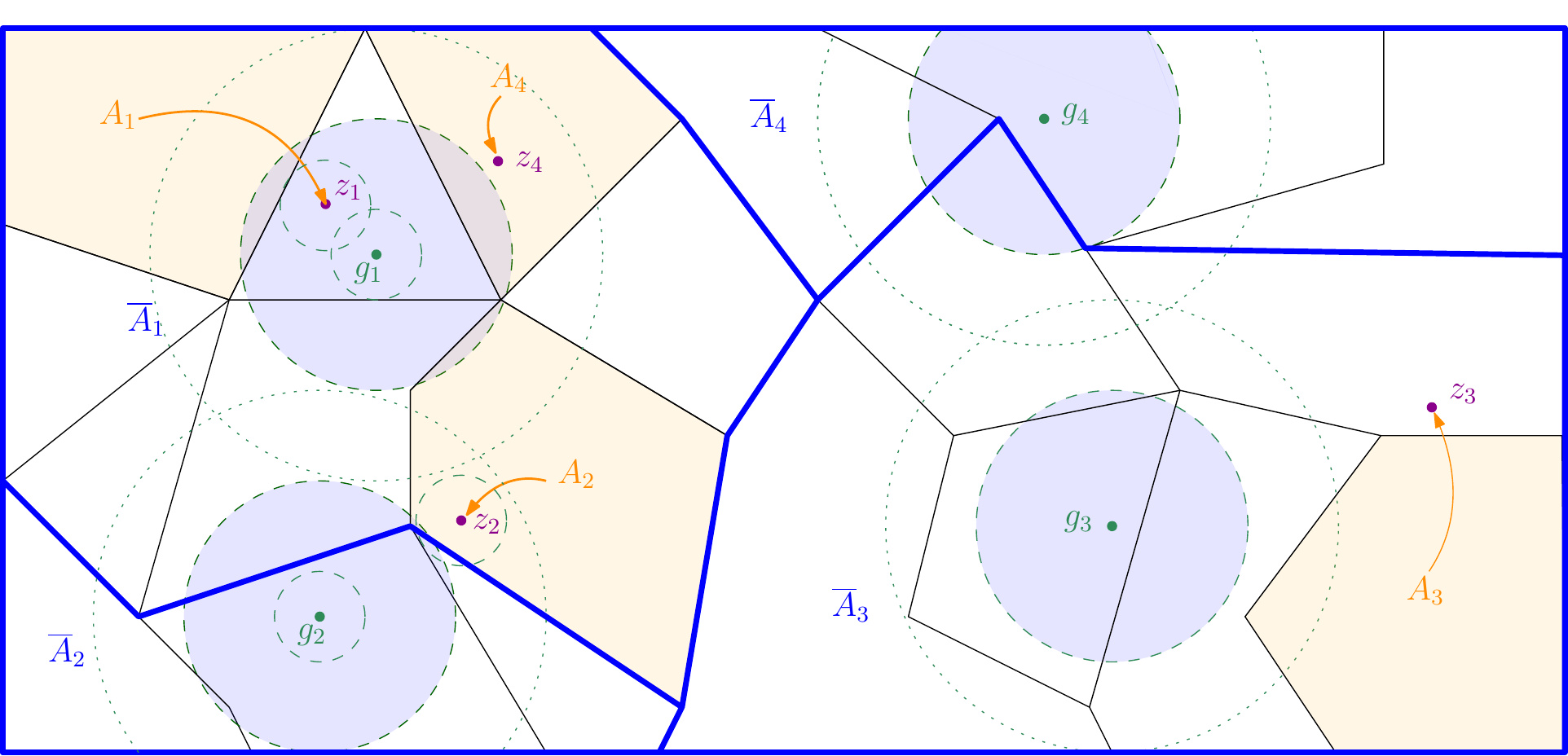}
    \caption{Visualization of an $(L,r)$-partition and of Algorithm~\ref{alg:construction_and_charging}.
The black polygons depict an $(L,r)$-partition $\mathcal{P}$ of the clients $C_{< \gamma}$. The blue polygons shows how the clients $C_{<\gamma}$ are partitioned by $\overline{\mathcal{P}}$.
Moreover, the blue $3r$-balls around $g_i$ illustrate which clients of $C_\gamma$ get assigned to the part $\overline{A}_i\in \overline{\mathcal{P}}$.
The dashed circles have radius $r$ and $3r$ respectively, while the dotted circles have radius $5r$.
    We assume $Z=\{z_i \mid i \in [4] \}$ is given and we construct the respective $\mathcal{\overline{A}}$ and $\overline{h}$, given $\mathcal{A}$ and $h$.
    We have $\mathcal{A}=\{ A_i \mid i \in [4] \}$ (the orange areas) and  $\overline{\mathcal{A}}=\{ \overline{A}_i \mid i \in [3] \}$. Moreover $h(A_i)=z_i$ for $i \in [4]$ (depicted by an orange arrow), while $\overline{h}(\overline{A}_i)=z_i$ for $i \in [3]$.
}
    \label{fig:main_lemma}
\end{figure}%
First, observe that $\overline{\mathcal{P}}$ is a partition.
It clearly covers all clients as no client is farther than distance $r$ away from its nearest facility, and we consider all facilities.
Moreover, the sets in $\overline{\mathcal{P}}$ are disjoint by construction.
Now, observe that any $\overline{A}_i \in \overline{\mathcal{P}}$ has small diameter, because
\begin{equation*}
\diam\left(\overline{A}_i\right) \leq 2\cdot \max_{c\in \overline{A}_i} d\left(g_i,c\right) \leq 10  r+2 L r \enspace ,
\end{equation*}
where the second inequality holds because $d(g_i,c)\leq 5r+Lr$ for any $c\in \overline{A}_i$ due to the following.
Consider $c\in \overline{A}_i$.
If $c\in C_\gamma$, then we even have $d(g_i,c) \leq 3r$.
Otherwise, let $A\in \mathcal{P}$ be the set in the partition $\mathcal{P}$ containing $c$.
Note that $c\in \overline{A}_i$ implies $A\subseteq \overline{A}_i$.
Hence, $d(g_i,c) \leq d(g_i,A) + \max\{d(b,c)\colon b\in A\} \leq 5r + Lr$, where we use $d(g_i,A)\leq 5r$, because $A\subseteq \overline{A}_i$, and $\diam(A)\leq Lr$, which holds because $\mathcal{P}$ is an $(L,r)$-partition.
Thus, property~\ref{item:LrPart_dist} of the definition of an $(2L+10,r)$-partition (\cref{def:L_r_partition}) is fulfilled for $\overline{\mathcal{P}}$. 

It remains to show that property~\ref{item:LrPart_z} holds for a given selection $Z$. 
To this end, we use that $\mathcal{P}$ is an $(L,r)$-partition, which implies that there is a subfamily $\mathcal{A} \subseteq \mathcal{P}$ and a corresponding injection $h\colon \mathcal{A}\to Z$ fulfilling property~\ref{item:LrPart_z} of \cref{def:L_r_partition} for the $(\gamma-1)$-colorful space $(C_{<\gamma}\discup F, d, \widehat{w})$.
In the following we construct $\overline{\mathcal{A}}\subseteq \overline{\mathcal{P}}$ and $\overline{h}\colon \overline{\mathcal{A}}\to Z$ such that property~\ref{item:LrPart_z} of  \cref{def:L_r_partition} is satisfied for $\overline{\mathcal{A}}$ and $\overline{h}$.
At the same time when constructing $\overline{\mathcal{A}}$, we employ a careful charging argument that makes sure that $w_\gamma \big( \bigcup_{\overline{A}\in \overline{\mathcal{A}}} \overline{A} \big) \geq w_\gamma(B_{C}(Z,r))$, i.e., that the constructed $\overline{\mathcal{A}}$ covers at least as much as $Z$ of color $\gamma$.
For the remaining colors, we show that the new selection $\overline{\mathcal{A}}$ includes all of $\mathcal{A}$; formally, we show that for each $A\in \mathcal{A}$, there is an $\overline{A}\in \overline{\mathcal{A}}$ such that $A\subseteq \overline{A}$.
This, as well as $d(\overline{A}, \overline{h}(\overline{A})) \leq r$ for all $\overline{A} \in \overline{\mathcal{P}}$ and injectivity of $\overline{h}$, are proved later.

For $i\in [|F|]$, we define
\begin{equation*}
U_i \coloneqq C\setminus \bigcup_{t=1}^{i-1} \overline{A}_t
\end{equation*}
to be the clients that are ``uncovered'' at step $i$.
By the way \cref{algo:iterative_decomposition} selects $g_i$ in each iteration $i\in [|F|]$, we have
\begin{equation*}
w_\gamma \left(B_{C}\left( g_i,r \right) \cap U_i \right) \geq w_{\gamma}\left( B_{C} (f,r) \cap U_i \right) \quad \forall i\in \left[|F|\right] \text{ and }f\in F\enspace,
\end{equation*}
which we call the \emph{greediness property}.

We now describe the construction of $\overline{\mathcal{A}}$ and the charging scheme in detail.
We successively add sets $\overline{A}_i\in \overline{\mathcal{P}}$ to $\overline{\mathcal{A}}$, where the sets $\overline{A}_i$ are considered in increasing order of their index.
When adding a set $\overline{A}_i$ to $\overline{\mathcal{A}}$, we also perform two further steps: (i) we identify an element $f\in Z$ and set $\overline{h}(\overline{A}_i)=f$, and (ii) we mark $f$ as assigned to make sure that we never assign it again in the future (as $\overline{h}$ needs to be an injection).
For convenience, for $i\in [|F|]$ and $f\in Z$, we write $\textsc{Assign}(i, f)$ for performing these steps, i.e., adding $\overline{A}_i$ to $\overline{\mathcal{A}}$, setting $\overline{h}(\overline{A}_i)$ to $f$, and marking $f$ as assigned.

The charging argument charges the coverage of color $\gamma$ of $B_C(Z,r)$ against the $\gamma$-coverage in $\bigcup_{\overline{A} \in \overline{\mathcal{A}}} \overline{A}$.
Whenever we charge a set $Q\subseteq B_{C}(Z,r)$ against some subset $W \subseteq \bigcup_{\overline{A}\in \overline{\mathcal{A}}} \overline{A}$, we make sure that $w_\gamma(Q) \leq w_\gamma(W)$.
\cref{alg:construction_and_charging} shows our procedure to construct both $\overline{\mathcal{A}}$ and the desired injection $\overline{h}\colon \overline{\mathcal{A}}\to Z$ together with the charging argument. (See also \cref{fig:main_lemma}.)
\begin{algorithm}
\DontPrintSemicolon
\caption{Construction of $\overline{\mathcal{A}}$ and injection $\overline{h}\colon \overline{\mathcal{A}}\to Z$ together with charging argument.}\label{alg:construction_and_charging}
Mark all facilities in $Z$ as unassigned.\;
\For{$i=1$ \KwTo $|F|$}{
\begin{minipage}{0.88\linewidth}
\smallskip
\begin{itemize}
\setlength\itemsep{2mm}
\item \textbf{Rule 1} 
If there is an unassigned $f \in Z$ with $B_C(f,r)\cap B_C(g_i,r)\cap U_i\neq \emptyset$: 

$\textsc{Assign}(i,f)$.

\item \textbf{Rule 2} 
Else if there is an unassigned $f\in Z$ with $B_C(f,r)\cap B_C(g_i,3r)\cap U_i \neq \emptyset$:

$\textsc{Assign}(i,f)$ and charge $B_{C}(f,r) \cap U_i$ against  $B_{C}(g_i,r) \cap U_i$. 

\item \textbf{Rule 3}
Else if there is an $A\in \mathcal{A}$ such that $A\subseteq \overline{A}_i$:

$\textsc{Assign}(i, h(A))$ and
charge $B_{C}(h(A),r) \cap U_i$ against $B_{C}(g_i,r)\cap U_i$.

\end{itemize}
\end{minipage}
\smallskip

If $\textsc{Assign}$ was called, charge against themselves all points in $\overline{A}_i$ that have not been charged yet.\;
}
\end{algorithm}

We start by showing that $\overline{h}$ is an injection.
Suppose $f$ is assigned using Rule~1 or~2. Then $f$ was not assigned so far as we only assign unassigned facilities.
Now suppose $h(A)=\overline{h}(\overline{A}_i)$ is assigned using Rule~3. We claim that $h(A)$ is not assigned so far. 
Assume by the sake of deriving a contradiction that it was assigned in a previous iteration $j<i$.
It cannot have been assigned by Rule~3, since $h$ is injective. So assume it is was assigned by Rule~1 or~2. Hence, $g_j$ satisfies $B_{C}(g_j,3r) \cap B_{C}(h(A),r)\cap U_i \neq \emptyset$. This implies that $d(g_j,h(A)) \leq 4r$ and thus $A \subseteq \overline{A}_j$, which contradicts $A \subseteq \overline{A}_i$.

Moreover, $\overline{A}_i$ fulfills property~\ref{item:LrPart_shortAssign} of a $(2L+10,r)$-partition because of the following.
Let $f\in Z$ and $\overline{A}_i \coloneqq \overline{h}^{-1}(f)$, and we have to show that $d(\overline{A}_i,f)\leq r$.
Because $h(\overline{A}_i)=f$, we called at some point during \cref{alg:construction_and_charging} the procedure $\textsc{Assign}(i,f)$.
In both Rule~1 and Rule~2 we have $B_C(f,r)\cap B_C(g_i,3r)\cap U_i \neq \emptyset$, which implies that $\overline{A}_i$ contains a client in $B_C(f,r)$, as desired.
If $\textsc{Assign}(i,f)$ was called in Rule~3, then we have $h^{-1}(f)\subseteq \overline{A}_i$, which implies $d(\overline{A}_i, f)\leq d(h^{-1}(f),f)\leq r$ by the fact hat $\mathcal{P}$ is an $(L,r)$-partition.

It remains to show that $\overline{A}$ fulfills property~\ref{item:LrPart_coverage} of an $(2L+10,r)$-partition.
We first consider the last color (color $\gamma$) and show $w_\gamma(\bigcup_{\overline{A}\in\overline{\mathcal{A}}} \overline{A}) \geq w_\gamma(B_C(Z,r))$.
To this end, observe that the charging indeed charges clients in $B_C(Z,r)$ against clients in $\bigcup_{\overline{A}\in \overline{\mathcal{A}}} \overline{A}$.
We allow for charging a client in $B_C(Z,r)$ against more than one client in $\bigcup_{\overline{A}\in \overline{\mathcal{A}}} \overline{A}$.
However, no client in $\bigcup_{\overline{A}\in \overline{\mathcal{A}}} \overline{A}$ gets charged against more than once because in iteration $i$ we only charge against clients in $\overline{A}_i$, and the sets $\overline{\mathcal{A}}=\{\overline{A}_1,\ldots, \overline{A}_{|F|}\}$ form a partition of $C$.
Also note that we always charge clients of $B_C(Z,r)$ against clients of $\bigcup_{\overline{A}\in \overline{\mathcal{A}}} \overline{A}$ of at least the same $w_\gamma$-weight.
This is true whenever charging happens in Rule~2 or Rule~3, because of the greediness property, and holds trivially for all other charging operations, which only charge clients against themselves.
To conclude that $w_\gamma(\bigcup_{\overline{A}\in \overline{\mathcal{A}}}\overline{A}) \geq w_\gamma(B_C(Z,r))$, it remains to observe that all of $B_{C}(Z,r)$ gets charged against something. 

To this end, fix a facility $f\in Z$.
Consider an iteration $j$ of \Cref{alg:construction_and_charging} such that $B_C(g_j,3r)\cap U_j$ intersects $B_C(f,r)$. We claim that for each such iteration, either $\textsc{Assign}(j,f)$ is called, 
or $B_C(f,r)  \cap B_C(g_j,3r) \cap U_j$ is charged.
To prove the claim, suppose $f$ is not assigned in iteration $j$.
By \Cref{alg:construction_and_charging}, either Rule 1 or Rule 2 must have applied in this iteration $j$, as $f$ satisfies the condition of Rule 2.
Thus $\textsc{Assign}$ was called on $j$ and all points in $\overline{A}_j$ have been charged.
Now suppose the first case applies, i.e., $\textsc{Assign}(j,f)$ is called for some $j$. Then all of $B_C(f,r)\cap U_j$ is charged (and $B_C(f,r)\setminus U_j$ is already charged by the second case).
If the first case never applies, then all of $B_C(f,r)$ is charged by the second case since $U_{|F|}$ is empty.
Hence, all of $B_C(f,r)$ is charged, as desired.

To see that property~\ref{item:LrPart_coverage} of  \cref{def:L_r_partition} is fulfilled also for all colors $\ell \in [\gamma -1]$, observe that
Rule 3 makes sure that any component that was in $\mathcal{A}$ will still be selected in $\overline{\mathcal{A}}$. 
Thus, $w_{\ell}(\overline{\mathcal{A}}) \geq w_{\ell}(B_C(Z,r))$ for all colors $\ell \in [\gamma]$.

It remains to show that $d(\overline{A}_i,\overline{h}(\overline{A}_i)) \leq r$.
If Rule~1 or Rule~2 is applied, this is satisfied as there is a client $c \in B_{C}(g_i,3r) \cap B_{C}(\overline{h}(\overline{A}_i),r) \cap U_i$;
because $c \in \overline{A}_i$ by construction, we have $d(\overline{h}(\overline{A}_i), \overline{A}_i) \leq d(\overline{h}(\overline{A}_i), c) \leq r$. 
If Rule~3 is applied for $A \subseteq \overline{A}_i$, we also have $d(\overline{h}(\overline{A}_i), \overline{A}_i) \leq d(\overline{h}(\overline{A}_i), A) = d(h(A), A) \leq r$, where the last inequality follows from $\mathcal{P}$ being an $(L,r)$-partition.
\end{proof}%
\cref{lem:main_part} now follows readily from \cref{lemma:part_inductionstep}.%
\begin{proof}[Proof of Lemma~\ref{lem:main_part}]
The proof follows by induction on $\gamma$. 
For the induction start, consider $\gamma=0$. The set $\{\{c\} \colon c\in C\}$ is a $(0,r)$-partition on every $0$-colorful space $\left(C\discup F, d, w\right)$.
The induction step is given by Lemma~\ref{lemma:part_inductionstep}. Note that $2\left(10(2^{\gamma-1}-1)\right)+10 = 10(2^\gamma -1)$.
The running time is clearly $O(\poly(|X|,\gamma))$ as every step in the induction takes time $O(\poly(|X|,\gamma))$.\footnote{As briefly mentioned earlier, one can obtain slightly better constants, leading to existence and conductibility of $(8\cdot 2^\gamma-10, r)$-partitions. This can be achieved by using $\gamma=1$ as base case, for which our techniques can be shown to imply that there are $(6,r)$-partitions. In the interest of simplicity, we use the slightly weaker bound in \cref{lem:main_part}.}
\end{proof}

\printbibliography 

@book{welsh2010matroid,
  title={Matroid theory},
  author={Welsh, D.J.A.},
  year={2010},
  publisher={Courier Corporation}
}

@article{CGM92,
title = {Random pseudo-polynomial algorithms for exact matroid problems},
journal = {Journal of Algorithms},
volume = {13},
number = {2},
pages = {258-273},
year = {1992},
issn = {0196-6774},
doi = {https://doi.org/10.1016/0196-6774(92)90018-8},
url = {https://www.sciencedirect.com/science/article/pii/0196677492900188},
author = {Camerini, P.M. and Galbiati, G. and Maffioli, F.}
}

@article{PW70,
  title={On the Vector Representation of Matroids},
  author={Piff, M. J. and Welsh, D. J. A.},
  journal={Journal of The London Mathematical Society-second Series},
  year={1970},
  pages={284-288}
}

@inproceedings{BIPV19,
  author    = {Bandyapadhyay, S. and
               Inamdar, T. and
               Pai, S. and
               Varadarajan, K. R.},
  editor    = {Bender, M. A. and
               Svensson, O. and
               Herman, G.},
  title     = {A Constant Approximation for Colorful k-Center},
  booktitle = {27th Annual European Symposium on Algorithms, {ESA} 2019, September
               9-11, 2019, Munich/Garching, Germany},
  series    = {LIPIcs},
  volume    = {144},
  pages     = {12:1--12:14},
  publisher = {Schloss Dagstuhl - Leibniz-Zentrum f{\"{u}}r Informatik},
  year      = {2019},
  url       = {https://doi.org/10.4230/LIPIcs.ESA.2019.12},
  doi       = {10.4230/LIPIcs.ESA.2019.12},
  bibsource = {dblp computer science bibliography, https://dblp.org}
}

@inproceedings{CKMN01,
  author    = {Charikar, M. and
               Khuller, S. and
               Mount, D. M. and
               Narasimhan, G.},
  editor    = {Rao Kosaraju, S.},
  title     = {Algorithms for facility location problems with outliers},
  booktitle = {Proceedings of the Twelfth Annual Symposium on Discrete Algorithms,
               January 7-9, 2001, Washington, DC, {USA}},
  pages     = {642--651},
  publisher = {{ACM/SIAM}},
  year      = {2001},
  url       = {http://dl.acm.org/citation.cfm?id=365411.365555},
  bibsource = {dblp computer science bibliography, https://dblp.org}
}

@article{CLLW16,
  author    = {Chen, D. Z. and
               Li, J. and
               Liang, H. and
               Wang, H.},
  title     = {Matroid and Knapsack Center Problems},
  journal   = {Algorithmica},
  volume    = {75},
  number    = {1},
  pages     = {27--52},
  year      = {2016},
  url       = {https://doi.org/10.1007/s00453-015-0010-1},
  doi       = {10.1007/s00453-015-0010-1},
  bibsource = {dblp computer science bibliography, https://dblp.org}
}

@Article{AAKZ21,
author={Anegg, G.
and Angelidakis, H.
and Kurpisz, A.
and Zenklusen, R.},
title={A technique for obtaining true approximations for k-center with covering constraints},
journal={Mathematical Programming},
year={2021},
issn={1436-4646},
doi={10.1007/s10107-021-01645-y},
url={https://doi.org/10.1007/s10107-021-01645-y}
}

@Article{JSS21,
author={Jia, X.
and Sheth, K.
and Svensson, O.},
title={Fair colorful k-center clustering},
journal={Mathematical Programming},
year={2021},
issn={1436-4646},
doi={10.1007/s10107-021-01674-7},
url={https://doi.org/10.1007/s10107-021-01674-7}
}

@misc{IV21,
      title={Non-Uniform $k$-Center and Greedy Clustering}, 
      author={Inamdar, T. and Varadarajan, K.},
      year={2021},
      eprint={2111.06362},
      archivePrefix={arXiv},
      primaryClass={cs.DS}
}

@article{CN19,
  author    = {Chakrabarty, D. and
               Negahbani, M.},
  title     = {Generalized Center Problems with Outliers},
  journal   = {{ACM} Trans. Algorithms},
  volume    = {15},
  number    = {3},
  pages     = {41:1--41:14},
  year      = {2019},
  url       = {https://doi.org/10.1145/3338513},
  doi       = {10.1145/3338513},
  bibsource = {dblp computer science bibliography, https://dblp.org}
}

@article{HPST19,
  author    = {Harris, D. G. and
               Pensyl, T. W. and
               Srinivasan, A. and
               Trinh, K.},
  title     = {A Lottery Model for Center-Type Problems With Outliers},
  journal   = {{ACM} Trans. Algorithms},
  volume    = {15},
  number    = {3},
  pages     = {36:1--36:25},
  year      = {2019},
  url       = {https://doi.org/10.1145/3311953},
  doi       = {10.1145/3311953},
  bibsource = {dblp computer science bibliography, https://dblp.org}
}

@inproceedings{CKLV17,
  author    = {Chierichetti, F. and
               Kumar, R. and
               Lattanzi, S. and
               Vassilvitskii, S.},
  editor    = {Guyon, I. and
            von Luxburg, U. and
               Bengio, S. and
               Wallach, H. M. and
               Fergus, R. and
               Vishwanathan, S. V. N. and
               Garnett, R.},
  title     = {Fair Clustering Through Fairlets},
  booktitle = {Advances in Neural Information Processing Systems 30: Annual Conference
               on Neural Information Processing Systems 2017, December 4-9, 2017,
               Long Beach, CA, {USA}},
  pages     = {5029--5037},
  year      = {2017},
  url       = {https://proceedings.neurips.cc/paper/2017/hash/978fce5bcc4eccc88ad48ce3914124a2-Abstract.html},
  bibsource = {dblp computer science bibliography, https://dblp.org}
}

@inproceedings{BCCN21,
  author    = {Bajpai, T. and
               Chakrabarty, D. and
               Chekuri, C. and
               Negahbani, M.},
  editor    = {Bansal, N. and
               Merelli, E. and
               Worrell, J.},
  title     = {Revisiting Priority k-Center: Fairness and Outliers},
  booktitle = {48th International Colloquium on Automata, Languages, and Programming,
               {ICALP} 2021, July 12-16, 2021, Glasgow, Scotland (Virtual Conference)},
  series    = {LIPIcs},
  volume    = {198},
  pages     = {21:1--21:20},
  publisher = {Schloss Dagstuhl - Leibniz-Zentrum f{\"{u}}r Informatik},
  year      = {2021},
  url       = {https://doi.org/10.4230/LIPIcs.ICALP.2021.21},
  doi       = {10.4230/LIPIcs.ICALP.2021.21},
  bibsource = {dblp computer science bibliography, https://dblp.org}
}

\appendix 

\section{Proof of Lemma~\ref{lem:hardness}}\label{app:hardness}

\begin{proof}[Proof of \cref{lem:hardness} (cf. \cite{JSS21})]

Consider an instance of XWB on a ground set $\overline{F}$ with a weight function $\overline{w}$, a target weight $\overline{m}$ and let the set of independent sets of the matroid be $\mathcal{I}$.

We construct an instance of $2$-Colorful Matroid Supplier. Let $\overline{W}$ be the maximal weight that occurs in the XWB instance, i.e. $\overline{W}=\max_{\overline f \in \overline F} \overline{w}(\overline f)$. 

For every $\overline f \in \overline F$, introduce a client $c$ with weight $w_1(c)= \overline{w}(\overline f)$ and $w_2(c)= \overline{W} - \overline{w}(\overline f)$
We distribute the clients on the line with distance $D$ and introduce one facility $f$ per client at the exact same position. We call $f$ associated with $\overline f$, when $\overline f$ is the facility responsible for introducing the corresponding client $c$.
Thus if we define a set $S$ of facilities to be feasible if the set of associated elements $\overline{S}$ is in $\mathcal F$. Thus we get an instance of $2$-Colorful Matroid Supplier setting.
Lastly, we define the covering constraints, $m_1= \overline{m}$, $m_2=\rk(\overline{F}) \cdot \overline{W} - \overline{m}$.

Without loss of generality, we may assume a solution to this clustering instance is of maximal cardinality, i.e., a basis in the matroid.
Now a subset of the facilities with radius $0$ is a valid clustering solution if and only if the associated set is a solution to XWB.

Assume for contradiction that there is an approximation algorithm $\mathcal A$ for 2-Colorful-$\mathcal F$-supplier with a finite approximation guarantee. Applying $\mathcal A$ to this instance and letting $D\to \infty$, this means that $\mathcal A$ can decide whether there is a radius $0$ solution or not and thus solve XWB.
\end{proof}

\section{A 7-approximation algorithm for Colorful Knapsack Supplier}
\label{app:knapsack}

In this section, we show how ideas of the algorithm of \cite{JSS21} for $\gamma$-Colorful $k$-Center can be leveraged to design a $7$-approximation algorithm for the knapsack center setting, thus achieving an approximation guarantee independent of the number of colors. (However, for the running time to be polynomial, we need $\gamma=O(1)$.)

\subsection{Sketch of modifications}
Note that several key changes are necessary. In particular, \cite{JSS21} uses the pseudo-approximation of \cite{BIPV19} for Colorful $k$-Center, which returns an infeasible solution with $k+\gamma-1$ facilities. In order to get feasible solutions, the first step of the algorithm of \cite{JSS21} is to check ($\gamma -1$ times) if two optimal clusters can be replaced by a single stretched cluster. Then, using the pseudo-approximation on the remaining sub-instance results in a feasible solution. Thus, they may assume that facilities of an optimal solution are well-separated, i.e. that ``saving'' optimal facilities in this way is not possible.

However, with a knapsack constraint on the facilities, the potential violation of the pseudo-approximation cannot, in general, be overcome with this method.
To deal with this, we use a more refined approach to make sure that the cost savings gained in the first step still compensate for the loss later incurred by the pseudo-approximation.
Importantly, we cannot assume that optimal facilities are well-separated, thus further adjustments are necessary: instead of being able to solve the well-separated and non-well-separated cases separately as in \cite{JSS21}, we deal with a combined case where we distinguish between ``cheap'' and ``expensive'' facilities, where only the expensive facilities are well-separated.
We can then apply the ideas of \cite{JSS21} in this restricted setting. 

In the final LP-rounding step where the pseudo-approximation of \cite{BIPV19} is used, we now have to carefully distinguish between two different types of clusters.
We can use the approach of \cite{JSS21} to bound the contribution of one type of cluster and may thus round fractional values on them down. Fractional values on the other type of cluster can be rounded up due to the initial cost savings.

We shall now discuss each of the modifications in detail.

\subsection{Phase 0: gaining cost-savings with Cost-Guessing}

Let $\mathcal{I}$ be an instance of $\gamma$-Colorful Knapsack Supplier, i.e. an instance of $\gamma$-Colorful $\mathcal{F}$-Supplier where $\mathcal{F}= \{S\subseteq F \mid \kappa(S)\leq K\}$ for some cost function $\kappa \colon F\to \mathbb{R}_{\geq 0}$ and budget $K\in \mathbb{R}_{\geq 0}$.
Let $OPT\subseteq F$ be the set of facilities of an optimal solution and let $r$ be the optimal radius. Throughout this part we assume that the optimal radius $r$ is known. This is possible since one can iterate over all possible client facility distances. 
We build a feasible solution $S$ of radius $7r$ by splitting up the instance into parts, such that on each part, there is a partial $7r$-solution which is at least as good as $OPT$ on this part.

In Phase $0$, we iteratively guess the heaviest $OPT$ facility that is close to another $OPT$ facility (in a precise sense to be defined later). Opening the cheaper one with sufficient radius to cover everything both facilities cover allows us to ``save'' the cost of the more expensive facility. After $\gamma$ many such guesses, we let $\sigma$ be the smallest possible cost saved this way. The upshot is that all remaining facilities in $OPT$ whose cost exceeds $\sigma$ must be ``separated'' from all other $OPT$ facilities. This is formalized below.

\begin{definition}[well-separated]
We call $f,g\in F$ well-separated if $d(f,g)>4r$.
\end{definition}
We first execute \cref{alg:weight_guessing}. Note that there are $|F|^{O(\gamma)}$ possible guesses.
We let $S_\kappa$ be the first partial solution and set $\sigma = \min_{e\in \overline{E}} \kappa(e)$. (If at some point in the  \textsc{for}-loop of \cref{alg:weight_guessing} no pair $(e,s)$ can be found, set $\sigma=0$.)

\begin{algorithm}[ht]
 $S_\kappa \leftarrow \emptyset$\;
 $\overline{E} \leftarrow \emptyset$\;
 $D \leftarrow \emptyset$\; 
 \For{$\ell\leftarrow 1$ \KwTo $\gamma$}{
Guess a pair $(e,s)$ of non-well-separated facilities in $OPT\setminus D$ such that $\kappa(e)$ is maximal\;
 $S_\kappa \leftarrow S_\kappa \cup \{s\}$\;
 $\overline{E} \leftarrow \overline{E} \cup \{e\}$\;
 $D\leftarrow D\cup B_F(s,4r)$\;
 }
 \KwRet{$S_\kappa,\overline{E}, D$}
 \caption{Cost-Guessing on $\mathcal{I}$}
 \label{alg:weight_guessing}
\end{algorithm}
Now consider a new instance $\mathcal{I}_\kappa$ of $\gamma$-Colorful Knapsack Center where we have removed this partial solution $S_\kappa$ with respect to the stretched radius $5r$, i.e., in the new instance $\mathcal I _{\kappa}$, the sets $C, F$, coverage requirements $m_\ell$ for $\ell\in [\gamma]$, and knapsack budget $K$ have been adjusted as follows:
\begin{align*}
m_\ell&\leftarrow m_\ell-w_\ell(B_C(S_\kappa,5r)) \quad \forall \ell\in [\gamma]\\
C&\leftarrow C \setminus B_C(S_\kappa,5r)\\
F &\leftarrow F\setminus D\\
K&\leftarrow K-\kappa(S_\kappa)-\gamma\sigma\enspace .
\end{align*}

Let $OPT_\kappa = OPT\setminus D$ be the remaining $OPT$ facilities.

\medskip
The crucial information we now have about $\mathcal I _\kappa$ is that all $OPT_\kappa$ facilities above the savings threshold $\sigma$ are, in fact, well-separated from all other $OPT_\kappa$ facilities. This is formalized in \cref{lem:well_separated} below.

\begin{definition}[expensive]
Call $f\in F$ expensive if $\kappa(f)>\sigma$ and let $E\subseteq F$ be the set of expensive facilities.
\end{definition}

\begin{lemma}
\label{lem:well_separated}
$OPT_\kappa$ is a feasible solution to $\mathcal I_\kappa$.
Furthermore, if $f\in OPT_\kappa$ is expensive, then $f$ is well-separated from every other element of $OPT_\kappa$.
\end{lemma}
\begin{proof}
Consider the cost and contributions of $OPT\setminus OPT_\kappa = OPT\cap D$.
The subset $B_C(OPT\cap D,r)\subseteq B_C(D,r)$ is fully covered by $B_C(S_\kappa,(4+1)r)$. Moreover, $\kappa(S_\kappa) +  \sigma \gamma \leq \kappa(OPT\cap D)$. This $OPT_\kappa$ is feasible for $\mathcal{I}_\kappa$.

Let $f\in OPT_\kappa$ be expensive. By the greedy selection of pairs in \cref{alg:weight_guessing}, $f$ must be well-separated from all other $OPT_\kappa$ facilities.
\end{proof}

\subsection{Phase 1: gaining color-coverage with Weight-Guessing}
\label{subsec:weight_guessing}

We may now apply Phase 1 of the approach of \cite{JSS21} on the expensive facilities only. Since we know they are well-separated from all other $OPT_\kappa$ facilities, we may stretch them and ``gain'' the weight of extra clients thus covered (which have not been covered by the optimal solution, by well-separatedness).
Note that we adapt the terms and definitions of \cite{JSS21} to the supplier setting with general weights on the clients. The arguments can be modified straightforwardly; we include the proofs and the procedure for completeness.

Define the flower of a client $c\in C$ as
\begin{equation*}
\F(c)=\displaystyle \bigcup_{f\in B_F(c,r)} B_C(f,r)\enspace,
\end{equation*} 
and, for $f\in B_F(c,r)$, we further define
\begin{equation*}
\gain_{\ell}(f,c)=  w_{\ell} \left(\F(c)\setminus B_C(f,r)\right) \qquad \text{ for all } \ell \in [\gamma].
\end{equation*}
In Phase 1 we guess, iteratively for each color $\ell\in [\gamma]$, up to $3\gamma$ many $OPT_\kappa$ facilities where the highest weight for a given color can be gained by stretching it by a factor of $3$. This is formalized in \cref{alg:color_guessing}.
Let $S_w$ be the facilities selected in this procedure, and let $\tau_\ell$ be the smallest weight gained for color $\ell$. 
\begin{algorithm}[ht]
 $S_w \leftarrow \emptyset $\;
 \For{$\ell \leftarrow 1$ \KwTo $\gamma$}{
    \For{$j\leftarrow 0$ \KwTo $3 \gamma$}{
        Guess a pair $f,c$ with $f \in E \cap \left(OPT_\kappa \setminus S_w\right)$, and $\gain_\ell(f,c)$ maximal\;
        $\tau_\ell \leftarrow \gain_\ell(f,c)$\;
        $S_w\leftarrow S_w\cup \{f\}$\;}
    }
\KwRet{$S_w$}
 \caption{Color-Guessing on $\mathcal I _w$}
 \label{alg:color_guessing}
\end{algorithm}
If for some $\ell \in [\gamma]$, no suitable pair $f,c$ exists in the inner \textsc{for-loop} of \cref{alg:color_guessing}, set $\tau_\ell =0$.

Now consider a new instance $\mathcal{I}_{\kappa w}$ of $\gamma$-Colorful Knapsack Center where we have removed the partial solution $S_w$, i.e., in the new instance $\mathcal I _{\kappa w}$, the coverage requirements $m_\ell$, facilities $F$, and knapsack budget $K$ have been adjusted as follows:
\begin{align*}
m_\ell&\leftarrow m_\ell-w_\ell(B_C(S_w,r))  \quad \forall \ell\in [\gamma]\\
C&\leftarrow C \setminus B_C(S_w,3r)\\
F &\leftarrow F\setminus S_w\\
K&\leftarrow K-\kappa(S_w) \enspace .
\end{align*}
Let $OPT_{\kappa w} = OPT_\kappa \setminus S_w$ be the remaining facilities of the optimal solution $OPT_{\kappa}$.

As in \cite{JSS21}, in this instance it is enough to find a solution $S_{\kappa w}$ that covers only $m_\ell-3\tau_\ell$ for each $\ell \in [\gamma]$, since then, together with opening $S_w$ with radius $3r$, this gives us a (stretched) solution to $\mathcal I_\kappa $.
Furthermore, this instance has the crucial property that the potential gain of expensive facilities is bounded by $\tau_\ell$.

\begin{lemma}[{\cite[Section~2.4]{JSS21}}]\label{lem:gain_limited}
$OPT_{\kappa w}$ is a feasible solution to $\mathcal{I}_{\kappa w}$. 
Moreover for all $f\in OPT_{\kappa w} \cap E$ and all $\ell \in [\gamma]$ where the $\max$ below exists, we have
\begin{equation*}\label{eqn:flower_condition}
\max_{c\in B_C(f,r)} \gain_\ell(f,c)\leq \tau_\ell\enspace.
\end{equation*}\
\end{lemma}

\begin{proof}

To prove feasibility we first prove the following claim 

\begin{claim}[cf. {\cite[Lemma~2]{JSS21}}]
\label{claim:gain_disjoint}
The gained regions are disjoint from $OPT_\kappa$ balls. Formally, let $f,c$ be a pair guessed in Algorithm~\ref{alg:color_guessing}. Then $B_C(OPT_\kappa\setminus \{f\}, r) \cap \F(c) = \emptyset$.
\end{claim}
Assume for contradiction there is a facility $f' \in OPT_\kappa \setminus \{f\}$ such that $B_C(f',r) \cap \F(c) \neq \emptyset$.
We have, $d(f,f')\leq d(f,c) + d(f', c) \leq r + 3r \leq 4r$. This is a contradiction since $f$ is expensive and expensive facilities of the optimal solution $OPT_\kappa$ are well-separated. This proves the claim.

Feasibility of the instance $\mathcal{I}_{\kappa w}$ follows by \cref{claim:gain_disjoint}.
It is possible to remove the corresponding clients from the instance because they have not been covered in $OPT_{\kappa}$.

The last part of the lemma follows by definition of $\tau_\ell$ and the execution of Algorithm~\ref{alg:color_guessing}.
\end{proof}

\subsection{Phase 2: separating the dense part with Dense Clusters}\label{sec:dense}

As in \cite{JSS21}, in Phase 2 we want to separate regions which are ``dense'', i.e. which contribute a significant amount of weight in one color (in relation to the gain-threshold $\tau_\ell$).
In order to be able to do so, a crucial ingredient is \cref{lem:gain_limited}. However, because it only applies to expensive facilities, the construction of dense sets in \cite{JSS21} has to be modified carefully. The following definition will be convenient.

\begin{definition}[$\beta$-expensive]
Call $f\in F$ $\beta$-expensive if $B_F(f,\beta  r)\subseteq E$ and let $E_{\beta}\subseteq F$ be the set of $\beta$-expensive facilities.
\end{definition}
The modified definition of dense sets is now as follows.

\begin{definition}[dense,~cf.~{\cite[Definition~4]{JSS21}}]
Call $f\in F$ dense on $C$ (with respect to $\ell\in [\gamma]$) if $f\in E_4$ and $w_\ell(B_{C}(f,r)) >2 \tau_\ell$.\\
If $f$ is dense, define the \emph{core} $\Core(f)$ of $f$ to be
\begin{equation*}
\Core(f)= \left\{g \in F  \, \middle \vert \,  w_{\overline \ell}(B_{C}(g,r)\cap B_{C}(f,r)) > \tau_{\overline \ell} \ \text{   for some  } \overline \ell \in [\gamma] \right\}
\end{equation*}
and define the cluster $\Cluster(f)$ to be
\begin{equation*}
\Cluster(f)= B_C(\Core(f),r)\enspace.
\end{equation*}
\end{definition}

We can now execute \cref{alg:dense_sets}, which returns the ``dense part'' $(C^d, F^d)$ of the instance.

\begin{algorithm}[ht]
$\mathcal{U}\leftarrow \emptyset $\;
 $C^d \leftarrow \emptyset$\;
 $F^d \leftarrow \emptyset$\;
 \For{$\ell \leftarrow 1$ \KwTo $\gamma$}{
 \While{there is a point $f\in F^s$ dense on $C\setminus C^d$ (with respect to $\ell$)}{
 $\mathcal{U}\leftarrow \mathcal{U}\cup \{(\Cluster(f)\setminus C^d, \Core(f)\setminus F^d)\}$\;
$C^d \leftarrow C^d \cup \{  \Cluster(f)\setminus C^d  \}$\;
$F^d \leftarrow F^d \cup \{ \Core(f)\setminus F^d  \}$\;
} 
 }
 \KwRet{$\mathcal{U}, C^d, F^d$}
 \caption{Dense Sets}
 \label{alg:dense_sets}
\end{algorithm}
The key property of the dense part is the lemma below, which states the optimal solution is cleanly separated by the dense part, i.e., for every $g \in OPT_{\kappa w}$, either $g$ is itself in $F^d$, or the ball around $g$ does not intersect the dense clients at all.

\begin{lemma}[cf. {\cite[Lemma~3]{JSS21}}]
\label{lem:dense_opt_intersection}
For any $g \in OPT_{\kappa w}$ exactly one of the following holds
\begin{enumerate}
\item\label{cond1} $g\in  F^d$, or
\item\label{cond2} $B_{C}(g,r) \cap C^d=\emptyset$\enspace.
\end{enumerate}
\end{lemma}

\begin{proof}
We will show that Condition~\ref{cond1} is equivalent to the negation of Condition~\ref{cond2}.
First assume~\ref{cond1} holds, i.e., $g \in \Core(f) \subseteq F^d$ for some $f\in F^d$. Then $B_C(g,r) \subseteq \Cluster(f) \subseteq C^d$, thus $B_C(g,r) \cap C^d \neq \emptyset$.

Now assume that Condition~\ref{cond2} does not hold.
Suppose first that $B_C(g,r)$ intersects $B_C(f,r)$ for a dense facility $f$, say $c\in B_C(f,r)\cap B_C(g,r)$. Then $g\in E$ since $f\in E_4$.
Then, since we know that all expensive elements of $OPT_{\kappa w}$ have limited gain, we know that $\gain_\ell (g,c) \geq w_\ell(B_{C}(f,r)\setminus B_{C}(g,r))$ can be at most $\tau_\ell$. Hence $w_\ell(B_{C}(g,r)\cap B_{C}(f,r))\geq \tau_\ell$, i.e., $g\in \Core(f)$ and thus $g\in F^d$.
Suppose now that $B_C(g,r)$ does not intersect $B_C(f,r)$ for any dense facility. Since Condition~\ref{cond2} does not hold, there must be a dense $f$ and some $h\in \Core(f)$ such that $B_C(g,r)\cap B_C(h,r)\neq \emptyset$, say $c\in B_C(h,r)\cap B_C(g,r)$. Then $d(g,f)\leq d(g,h)+d(h,f) \leq 4r$, thus $f\in E_4$ again implies $g\in E$.
Moreover, $B_C(f,r)\cap B_C(h,r)$, which contains more than $\tau_\ell$ for some color $\ell\in [\gamma]$, is contained in $\F(c)\setminus B_C(f,r)$. Thus $\gain(g,c)> \tau_\ell$, which is a contradiction as $g$ is expensive.
\end{proof}

This allows us to efficiently recover a solution $S_d$ that is guaranteed to be at least as good as $OPT_{\kappa w}\cap F^d$, i.e. the optimal solution on the dense part.

\begin{lemma}[cf. {\cite[Lemma~4]{JSS21}}]
\label{lem:dense}
We can efficiently find a radius-$5r$ solution $S_d$ such that
\begin{itemize}
    \item $\kappa(S_d) \leq \kappa(OPT_{\kappa w} \cap F^d)$
    \item $w_{\ell} (B_{C^d} (S_d,5r) ) \geq w_{\ell}( B_{C^d}(OPT_{\kappa w} \cap F^d),r)$ for all $\ell \in [\gamma]$.
\end{itemize}
\end{lemma}
\begin{proof}
From \cref{lem:dense_opt_intersection} we have that
$\{f \in OPT_{\kappa w} : f \in F^d \}= \{f \in OPT_{\kappa w} : B(f,r) \cap C^d = \emptyset \}$. 
Thus, the dense part is cleanly separated from the remaining instance. This allows us to use a a dynamic program (DP).

Let $\mathcal{U}$ be as returned by \cref{alg:dense_sets}. For $u=(C', F')\in \mathcal{U}$, let $w_\ell(u)\coloneqq w_\ell(C')$ for all $\ell\in [\gamma]$ and $\eta(u)\coloneqq \min_{f\in F'} \kappa(f)$, i.e. we assign to $u$ the weight of the clients in the cluster $C'$ and the minimal weight of the facilities in the corresponding core $F'$.
Moreover we guess how much of color $\ell$ for $\ell \in [\gamma]$ is covered by an optimal solution on the dense part. Denote this by $m_{\ell}^d$ and note that this is possible in time $|C|^{O(\gamma)}$ as $m_{\ell}^d \leq m_{\ell}  \leq |C|$.

Then the problem of finding a radius-$5r$ solution in the dense part can be formulated as the following binary problem. (Note that we achieve a $5r$ solution here, because any facility of the core  $F'$ could be selected to pay for the cluster $C'$.)
\begin{equation*}
\begingroup
\renewcommand*{\arraystretch}{1.5}
\begin{array}{r>{\displaystyle}rll@{\quad}l}
\min & \sum_{u\in \mathcal{U}} \eta (u) \cdot z(u) &                                           \\
     & \sum_{u\in \mathcal{U}} w_\ell(u)  \cdot z(u) & \geq  & m_\ell^d & \forall \ell\in [\gamma] \\
     & z                                             & \in   & \{0,1\}^{\mathcal{U}} \enspace.
\end{array}
\endgroup
\end{equation*}
The binary problem can be solved efficiently by standard dynamic programming techniques due to the unary encoding of the weights $w_\ell$ for $\ell\in [\gamma]$ (see, e.g.,~\cite{AAKZ21} for details).
We compute an optimal solution $z^*$ to the above binary program. This directly implies a solution $S_d$ by opening the facility determining the cost $\eta(u)$ for every $u \in \mathcal U$ with $z^*(u)=1$. 
Observe that $OPT_{\kappa w}\cap F^d$ corresponds to a solution to the binary program with objective value at most $\kappa(OPT_{\kappa w}\cap F^d)$. Hence $S_d$ satisfies $\kappa(S_d)\leq \kappa(OPT_{\kappa w})$ and $S_d$ also covers (with respect to radius $5r$) at least as much as the respective optimal solution, as claimed.
\end{proof}

\subsection{Flower-Polytope on sparse part}

Let $\mathcal{I}_{\kappa w d}$ be the instance remaining to be solved on the sparse part, i.e.
\begin{align*}
m_\ell&\leftarrow m_\ell-w_\ell(B_{C^d}(S_d,5r))  \quad \forall \ell\in [\gamma]\\
C&\leftarrow C \setminus C^d\\
F &\leftarrow F\setminus F^d\\
K&\leftarrow K-\kappa(S_d) \enspace .
\end{align*}

Let $OPT_{\kappa w d} = OPT_{\kappa w} \setminus F^d$.

We are left to find a suitable solution on $\mathcal{I}_{\kappa w d}$. By \cref{lem:well_separated}, we can afford to use an additional $\gamma \sigma$ in the Knapsack cost and by \cref{lem:gain_limited}, we can afford to only cover $m_\ell - 3\gamma \tau_\ell$ in each color $\ell$ (compared to $OPT_{\kappa w d}$).
Thus, the following lemma is sufficient for the final rounding step.

\begin{lemma}[cf. {\cite[Lemma~4]{JSS21}}]
\label{lem:sparse}
We can efficiently find a solution $S_s$ for $\mathcal{I}_{\kappa w d}$, such that
\begin{enumerate}
\item \label{item:loose_knapsack} $\kappa (S_s) \leq \kappa \left( OPT_{\kappa w d} \right) +\gamma \sigma $
\item \label{item:loose_covering} $w_{\ell} (B_{C} (S^s,7r) ) \geq w_{\ell} \left( B_{C}\left(OPT_{\kappa w d},r\right)\right) - 3\gamma \tau_\ell $ for all $\ell \in [\gamma]$.
\end{enumerate}
\end{lemma}

To prove \cref{lem:sparse}, we need to modify the arguments of \cite{JSS21}. In particular, we need to distinguish between different types of flowers when we round a fractional solution to the flower polytope. We defer the proof of \cref{lem:sparse} to the end of this section.

The following lemma allows us to identify certain clients which are not covered by the optimal solution and hence may be excluded.
We may do this for clients whose flowers have a high contribution and that are only surrounded by expensive facilities. 

\begin{lemma}[cf. {\cite[Lemma~4]{JSS21}}]
\label{lem:sparse_feasible}
$OPT_{\kappa w d}$ is feasible for $\mathcal{I}_{\kappa w d}$.
Moreover, for $c\in C$ such that $B_F(c,r)\subseteq E_4$ and $w_\ell(\F(c))>3\tau_\ell$ for some $\ell\in [\gamma]$, we have  $OPT_{\kappa w d}\cap B_F(c,r)=\emptyset$.
\end{lemma}
\begin{proof}
First observe that $OPT_{\kappa w d}$ is feasible for $\mathcal{I}_{\kappa w d}$ by \cref{lem:dense_opt_intersection}.
Next, let $c\in C$ be such that $B_F(c,r)\subseteq E_4$ and $w_\ell(\F(c))>3\tau_ell$ for some $\ell \in [\gamma]$.
Suppose there exists $f\in OPT_{\kappa w}\cap B_F(c,r)$. As $B_F(c,r)\subseteq E_4$, we have $f\in E_4$. Because we have removed dense sets, $f$ cannot be dense for any color, so we have $w_\ell(B_{C}(f,r))\leq 2\tau_\ell$. Thus, we must have $w_\ell( \F(c)\setminus B_C(f,r)) > \tau_\ell$, contradicting the fact that the gain of expensive facilities is bounded by $\tau_{\ell}$ for all $\ell \in [\gamma]$.
\end{proof}
We now consider the canonical relaxation of the $\gamma$-Colorful Knapsack Supplier problem $P$ below. By \cref{lem:sparse_feasible}, this remains feasible even if we add the constraint that clients $c\in C$ may \emph{not} be covered if $B_F(c,r)\subseteq E_4$ and $w_\ell(\F(c))>3\tau_\ell$ for some $\ell\in [\gamma]$ .
\renewcommand{\arraystretch}{1.3}
\begin{equation*}
P= \left\{ (x,y)\in [0,1]^C\times [0,1]^F \, \middle \vert \, 
\begin{array}{rclc}
\displaystyle \sum_{f\in F} \kappa(f)y(f)  & \leq & K \\
x(c) & \leq & y\left( B_F(c,r) \right ) & \forall c\in C \\
\displaystyle\sum_{c\in C} w_\ell(c)  x\left(c\right) & \geq & m_\ell & \forall \ell \in [\gamma]
\end{array}
 \right\} \enspace .
\end{equation*}
\renewcommand{\arraystretch}{1.0}
We can then take consider a point $(x,y)$ which is feasible for $P$ even with the added constraints that exclude clients.
Starting from such a feasible solution, we use the sparsification algorithm (\cref{alg:flower_lp_preparation}) of~\cite{HPST19} and~\cite{BIPV19} (modified to the weighted version and the supplier setting) to get a ``Flower-instance'' given by the sets $D_c$ for $c\in Q$.
\begin{algorithm}[ht]
 $Q \leftarrow \emptyset$\;
 $C' \leftarrow C$\;
 
 \While{$C' \neq \emptyset$ \textnormal{and} $\max_{b \in C'}x (b)>0$}{
    
        $c \leftarrow \argmax_{b\in C'} x(b)$\;
        $Q \leftarrow Q \cup \{c \}$\;
        $z_{c}\leftarrow \min \{ 1, y(B_{F}(c,r))$\}\;
        $D_{c} \leftarrow \F (c) \cap C'$\;
        $C' \leftarrow C'\setminus D_c$\;
        For all $b \in D_{c}$ set $\overline{x} (b) \leftarrow z (c)$
    }
\KwRet{$Q,  \{D_c \mid c\in Q\}, z, \overline{x}$}
 \caption{Flower-polytope preparation}
 \label{alg:flower_lp_preparation}
\end{algorithm}

Setting the weight of the (partial) flower $D_c$ of $c$ to be the smallest weight of any facility close to $c$, i.e., $\eta(c)=\min_{f\in B_F(c,r)} \kappa(f)$, we get the Flower-polytope $\overline{P}$ given below.

\renewcommand{\arraystretch}{1.3}
\begin{equation*}
\overline{P}= \left\{ z\in [0,1]^Q \, \middle \vert \, 
\begin{array}{rclc}
\displaystyle \sum_{c\in Q} \eta(c) z(c)  & \leq & K \\
\displaystyle \sum_{c\in Q}w_\ell (D_c ) z(c)   & \geq & m_\ell & \forall \ell \in [\gamma]
\end{array}
 \right\} \enspace .
\end{equation*}
\renewcommand{\arraystretch}{1.0}

\begin{lemma}[cf.~{\cite[Lemma~4]{JSS21}}]
\label{lem:flower_lp}
The Flower-polytope is non-empty.
\end{lemma}
\begin{proof}
We show that the vector $z\in [0,1]^Q$ returned by \cref{alg:flower_lp_preparation} is a feasible point in $\overline P$.

First observe that for any $b \in D_{c}$ we have 
$\overline{x} (b)=\overline{x} (c) \geq x (c) \geq x (b)$,
where the equality follows from \cref{alg:flower_lp_preparation}, the first inequality follows from the constraints of $P$, and the second inequality follows from the greedy choice of $c$.

To show feasibility of $z$ for $\overline{P}$, we first show that $\sum_{c \in Q} w_{\ell} (D_c)z(c)  \geq m_\ell$ for all $\ell \in [\gamma]$:
\begin{align*}
    \sum_{c \in Q} w_{\ell} (D_c) z(c)  &= \sum_{c \in Q} \sum_{b \in D_c} w_{\ell} (b) z(c)  \\
    &= \sum_{c \in Q} \sum_{b \in D_c} w_{\ell} (b) \overline{x}(b) \\
    & \geq \sum_{c \in Q} \sum_{b \in D_c} w_{\ell}(b)x(b) \\
    &= \sum_{b \in C} w_{\ell}(b) x(b)  \geq m_{\ell}\enspace.
\end{align*}
The first equality follows by definition of $w_{\ell}(D_c)$, the second equality follows from $z(c)= \overline x(b)$ for all $b \in D_c$, the inequality follows from the observation. Finally, we note that the sets $D_c$ are disjoint and cover all clients $b \in C$ with $x(b)>0$.

The second condition for feasibility is 
$\sum_{c \in Q} \eta (c)  z(c) \leq K$, which holds due to the following:
\begin{align*}
    \sum_{c \in Q} \eta (c)z(c)  &\leq \sum_{c \in Q} \eta(c) \sum_{f \in B_F(c,r)} y(f) \\
    & \leq \sum_{c \in Q} \sum_{f \in B_F(c,r)}  \kappa (f)y(f)\\
    & \leq \sum_{f \in F}  \kappa(f) y(f) \leq K\enspace.
\end{align*}
The first inequality follows by definition of $z$, the second one by the definition of $\eta (c)$, the third one by \cref{alg:flower_lp_preparation}, and the last one because $(x,y)\in P$.

\end{proof}
We can not prove \cref{lem:sparse}.
\begin{proof}[Proof of \cref{lem:sparse}]
We can efficiently find a vertex $z$ solution of $\overline{P}$.
By standard sparsity arguments, we obtain that $z$ has at most $\gamma+1$ many fractional entries because $\overline{P}$ has only $\gamma+1$ non-trivial constraints.
Moreover, by choosing $z$ to be an optimal vertex solution of $\min\{\sum_{c\in Q}\eta(c) y(c) \colon y\in \overline{P}\}$, we obtain that $z$ has at most $\gamma$ fractional entries, because the constraint $\sum_{c\in Q} \eta(c) z(c)\leq K$ can be dropped from $\overline{P}$ without changing the polytope, assuming that $\overline{P}\neq\emptyset$.

For each $c \in Q$ with $z(c)=1$, add $f_c =\argmin _{f\in B_F(c,r)} \kappa(f)$, i.e. the facility in $B_F(c,r)$ of minimal weight, to $S_s$.

For each $c\in Q$ such that $z(c)\in (0,1)$, do the following:
\begin{itemize}
\item if $B_F(c,r)\subseteq  E_4$, do not add anything to $S_s$;
\item if not, there exists an $f_c \in B_F(c,5r)\setminus E$. Add $f_c$ to $S_s$.
\end{itemize}
To prove that $S_s$ satisfies property~\ref{item:loose_knapsack}, note that
\begin{equation*}
\kappa(S_s) \leq  \sum_{c\in Q, z(c)=1} \eta(c) + \sum_{c\in Q, 0<z(c)<1} \kappa(f_c) \leq K + \gamma \sigma \enspace ,
\end{equation*}
since $f_c\in B_F(c,5r)\setminus E$ has $\kappa(f_c)\leq \sigma$ and there are at most $\gamma$ many fractional values in $z$.
To prove that property~\ref{item:loose_covering} of \cref{lem:sparse} is satisfied, note that $B_C(f_c,7r)$ contains $D_c$.
When $c\in Q$ is such that no $f_c$ is added to $S_s$, we lose the contribution of $D_c$. However, in this case $B_F(c,r)\subseteq E_4$ and therefore the contribution of $D_c\subseteq \F(c)$ is at most $3\tau_\ell$ for each color $\ell$ (by construction of the flower-polytope). Furthermore, there are at most $\gamma$ fractional values; hence we lose at most $3\gamma\tau_\ell$ for each color $\ell\in [\gamma]$.
\end{proof}
\begin{proof}[Proof of Theorem~\ref{thm:knapsack}]
The solution $S_{\kappa} \cup S_w \cup S_d \cup S_s$ is a feasible solution of radius $7r$ by the sequential application of \cref{lem:well_separated}, \cref{lem:gain_limited}, \cref{lem:dense}
 and \cref{lem:sparse}. Note that each of $S_\kappa, S_w, S_d, S_s$ can be constructed efficiently.
\end{proof}

\end{document}